\documentclass{svproc}
\usepackage{amsmath,graphicx}
\usepackage{amssymb}
\usepackage{algorithm}
\usepackage{algpseudocode}
\usepackage{hyperref}

\usepackage{url}

\newcommand{\C}{\mathbb{C}}
\newcommand{\R}{\mathbb{R}}
\newcommand{\Z}{\mathbb{Z}}

\newcommand{\Q}{\mathbb{Q}}

\newcommand{\fonc}[5]{\begin{aligned} #1 ~ : ~ #2 &\longrightarrow #3 \\ #4 &\longmapsto #5 \end{aligned}}

\DeclareMathOperator{\Trace}{\text{Trace}}
\DeclareMathOperator{\grevlex}{\text{grevlex}}
\DeclareMathOperator{\NF}{\text{NF}}
\DeclareMathOperator{\LT}{\text{LT}}
\DeclareMathOperator{\LM}{\text{LM}}
\DeclareMathOperator{\LC}{\text{LC}}

\DeclareMathOperator{\Res}{\text{Res}}
\DeclareMathOperator{\Disc}{\text{Disc}}
\DeclareMathOperator{\MT}{\text{MT}}
\DeclareMathOperator{\VTr}{\text{VTr}}

\begin{document}
\mainmatter              
\title{Solving parametric polynomial systems using Generic Rational Univariate Representation}
\titlerunning{Generic Rational Univariate Representation}  
%
\author{Florent Corniquel \inst{1}}
\authorrunning{Florent Corniquel} 
%
\tocauthor{Florent Corniquel}
\institute{Sorbonne Université, Université Paris-Cité, CNRS (IMJ-PRG), INRIA\\
\email{florent.corniquel@inria.fr}}

\maketitle              

\begin{abstract}
In this paper, we present a generic parametrization of generically zero-dimensional parametric polynomial systems. More specifically, we study the specialization properties of the Rational Univariate Representation and derive bounds on the degrees and heights of its elements. In addition to that, we propose two algorithms to effectively compute this parametrization.
\end{abstract}
\section{Introduction}
\subsection{Classical zero-dimensional case}

Consider $K$ a field with characteristic zero and $f_1, \dots, f_n \in K[X_1, \dots, X_n]$ polynomials that generate a zero-dimensional ideal $\mathcal{I}$. It is possible to find a parametrization of the affine variety $V(\mathcal{I})$ using Chow forms. Define the so called Chow form by  $h(U_0, \dots, U_n) := \prod_{x \in V(\mathcal{I})} (U_0 - x_1U_1 - \dots - x_n U_n)^{\mu(x)}$ where $U_0, \dots, U_n$ are independant variables and $\mu(x)$ denotes the multiplicity of $x$. We then chose coefficients $a_1, \dots, a_n \in K$ such that the linear form $t = a_1 X_1 + \dots + a_n X_n$ is injective on $V(\mathcal{I})$. In that case, we say that $t$ \textit{separates} $V(\mathcal{I})$. The initial system $f_1 = \dots = f_n = 0$ then becomes equivalent to the one given by 
$$\begin{cases} h(U_0, a_1, \dots, a_n) = 0 \\
 X_i = \frac{\frac{\partial h}{\partial U_i} (U_0, a_1, \dots, a_n)}{\frac{\partial h}{\partial U_0} (U_0, a_1, \dots, a_n)}, ~ 1 \leqslant i \leqslant n \end{cases}.$$ 

We denote by $h_0 := h(U_0, a_1,\dots, a_n)$, $h_1 := \frac{\partial h}{\partial U_0} (U_0, a_1, \dots, a_n)$ and $h_{X_i} := \frac{\partial h}{\partial U_i} (U_0, a_1, \dots, a_n)$ for $i \in [\![1,n]\!]$. Here, $[\![1,n]\!]$ is the set of integers between 1 and $n$. The family $\mathcal{F} := \{h_0, h_1, h_{X_1}, \dots, h_{X_n}\} \subseteq K[U_0]$ is called a \textit{Rational Univariate Representation} (RUR) of $V(\mathcal{I})$. 

Different approaches have been tried to compute these polynomials. Renegar in \cite{Renegar1992} expresses $h$ as a factor of Van der Waerden's $u$-resultant. The main issue with this approach is that the $u$-resultant can identically vanish, for example when the variety $V(\mathcal{I})$ has a component of dimension greater than zero or points at infinity. However, this problem can be solved using Canny's Generalized Characteristic Polynomial (\cite{Canny1990}). 

The most successful implementations have been based on the knowledge of the quotient $\mathcal{A} := K[X_1,\dots,X_n]/\mathcal{I}$, and more specifically through a basis $\mathcal{B} = \{p_1, \dots, p_D\}$ of $\mathcal{A}$ as well as the multiplication matrices $M_{X_1}, \dots, M_{X_n}$ by each variable in $\mathcal{A}$. These can be computed from a Gröbner basis of the initial system with $O(nD^3)$ operations in $K$ (see for example \cite{FGLM1993}). This approach is sharpened by Rouillier in \cite{Rouillier1998} where he shows that the elements of $\mathcal{F}$ can be computed using traces of multiplication matrices with regard to $\mathcal{B}$. The fact that the linear form $t = a_1 X_1 + \dots + a_n X_n$ separates $V(\mathcal{I})$ is crucial and the proposed algorithm is Las-Vegas : the linear form is chosen "randomly" and certified to be injective afterwards. 

Certifying that the form $t$ separates $V(\mathcal{I})$ is an important part of the problem which in most cases turns out to be as costly as the computations of the polynomials. As the polynomials $h_0, h_1, h_{X_1}, \dots, h_{X_n}$ can be computed without knowing whether $t$ separates $V(\mathcal{I})$ or not, the family $\mathcal{F}$ is rather called a \textit{Rational Univariate Representation-candidate} (RUR-candidate) in general, and specified to be a RUR when $t$ is separating. This distinction makes sense as the elements of $\mathcal{F}$ can be computed anyway. 

In \cite{DRR2025}, the authors propose a new Las-Vegas algorithm which given a zero-dimensional ideal $\mathcal{I}$ solves both the problem of computing a RUR of $V(\mathcal{I})$ and certifying that a given linear form $t$ is separating in $\tilde{O} (D^2 \delta + nD^2 (D - \delta + 1))$ arithmetic operations, from the knowledge of the multiplication matrices $M_{X_1}, \dots, M_{X_n}$ by the variables in the quotient algebra $K[X_1, \dots, X_n]/\mathcal{I}$. Here, $\delta$ is the degree of the minimal polynomial of $t$ and $D$ is the dimension of the quotient algebra.

\subsection{Parametric systems}
In the case of parametric systems, we consider $f_1, \dots, f_n \in K[W,X]$. We say that $W := [W_1, \dots, W_s]$ represents the parameters (which we want to replace by specific values) and that $X := [X_1, \dots, X_n]$ are the unknowns that we would like to express in terms of the parameters. Different approaches exist and we will briefly summarize them.

First, one can mention comprehensive Gröbner bases first introduced by Weispfenning in \cite{Weispfenning1992}. A comprehensive Gröbner basis can be understood as a family of polynomials which specializes to a Gröbner basis of the specialized system for all values of parameters. They are usually built from a Gröbner system, that is, a family of couple $(S_i,G_i)$ such that $S_i$ is a constructible subset of the parameters space and $G_i$ specializes to a Gröbner basis for all parameters values in $S_i$. The $S_i$'s are constructed to form a partition of the parameters space and the main drawback of the method is the size of the partition which tends to exponentially grow. This method comes with a very high complexity and it is not clear whether it can be effectively implemented.

In his Ph. D. thesis (\cite[Chapter 7]{Schost2001}), Eric Schost extends the use of triangular sets to solve parametric systems by working in $\Q(W)[X]$. He reaches a bound on the degree of the polynomials appearing in the triangular sets not better than $d^{o(n^2)}$ where $d$ is the maximum degree of the polynomials in the initial system, and does not try to implement the method. After that, Schost proposes in \cite{Schost2003} a probabilistic algorithm to compute a parametric resolution. The idea is to lift the solution of a specialized system using a formal Newton operator. The use of Straight-Line Programs allows him to reach a polynomial complexity in the size of the output. However, the proposed solution works when excluding the jacobian ideal, which means by excluding possible multiplicities.

In general, studying parametric systems is difficult and asking for certified pieces of information concerning the system often leads to costly computations. For example in the case of real geometry, Collins introduces in \cite{Collins1975} the Cylindrical Algebraic Decomposition (CAD) which is a partition of the space into cells where given polynomials have a constant sign. The complexity is known to be doubly exponential in the worst case with regard to the number of variables (see for example \cite{DH1988}). 

This technique can for example complement the computation of a discriminant variety of the space of parameters. This object introduced in \cite{LR2007} is a variety formed by the parameters values for which the specialized system has a "bad" behavior like roots at infinity or an infinite number of solutions. This variety has a very convenient property : its complement is a union of open subsets where the specialized system has a finite and constant number of solutions. The authors in \cite{LR2007} propose an algorithm to compute the discriminant variety and one can also compute a CAD of its complement to get an explicit description of the aforementioned open subsets. However, the complexity of computing a discriminant variety is also essentially exponential (see \cite{Moroz2006}).

Finally, we can mention the work done in \cite{LSeD2022} where the authors compute semi-algebraic formulas defining semi-algebraic subsets of the parameters space where the specialized system has a constant number of real solutions. Their approach relies on computing the matrix of the parametric Hermite's quadratic form and study its rank. 

\subsection{Contributions and outline of the article}

In Section 2, we introduce the notion of Generic Rational Univariate Representation (GRUR) and show that the Rational Univariate Representation (RUR) (in the sense introduced by Rouillier in \cite{Rouillier1998}) computed in $\C(W)$ has great specialization properties. Consider $t \in \C[X]$ a linear form that separates the specialized variety for almost all values of parameters. Then we show that except for an algebraic set of parameters with measure zero, the RUR computed in $\C(W)[X]$ specializes to a Rational Univariate Representation of the specialized variety thus justifying the name of \textit{Generic} RUR. 
For that purpose, we construct the parametric $u$-resultant of the system using Van der Waerden's approach and study its specialization behavior. 

Unfortunately, this approach is too costly complexity-wise. Indeed, even though the $u$-resultant can be computed from the Macaulay matrix as the quotient of two determinants (see \cite{Macaulay1902}), factorizing this resultant quickly becomes too hard. However, we can derive bounds for the degrees and height of the polynomials of the GRUR from this. We present it in Section 3 by making good use of some Arithmetic Nullstellensätz proved in \cite{DKS2012}. 

In Section 4 and 5, we propose two different algorithms to compute the GRUR. The first one relies on linear algebra and more precisely on the fact that the quotient algebra $\C(W)[X]/\mathcal{I}_W$, where $\mathcal{I}_W$ is the ideal generated by the equations in $\C(W)[X]$, is a vector space of finite dimension. This approach is based on the work of \cite{Rouillier1998}. The second algorithm is based on an evaluation/interpolation scheme. The idea is to compute as many specialized RURs as possible using the new algorithm proposed in \cite{DRR2025} and to interpolate the coefficients of the GRUR. If the initial equations have degree in $X$ bounded by $d_X$ and degree in $W$ bounded by $d_W$, we show that the GRUR can be computed using $\tilde{O}(sn^{2s+2} d_X^{2ns + 2n} d_W^{2s+1} + n^{s+1}d_X^{2ns + 5n} d_W^s)$ operations in $\Q$. We can then use the first polynomial of the computed GRUR to tackle the real roots classification problem and propose an algorithm to solve it whose complexity is higher than existing algorithms, but does not require any assumption on the system.

\section{Construction of the GRUR}

We start by introducing some notations.

\begin{itemize}
    \item $W:=[W_1,\dots, W_s]$ are parameters and $X:=[X_1,\dots,X_n]$ are unknowns ;
    \item given a value of parameters $w \in \C^s$, we denote by $\phi_w$ the specialization morphism that sends $W_i$ to $w_i$ ;
    \item we denote by respectively $<_X:=\grevlex(X)$ and $<_W:=\grevlex(W)$ the graded reverse lexicographic orderings in variables $X$ and $W$, and by $<_{X,W}:=\grevlex(W)<\grevlex(X)$ the associated elimination block ordering ;
    \item we will use the notation $\LT$, $\LM$ and $\LC$ indiced by "$<_X$", "$<_W$" or "$<_{W,X}$" for the leading term, monomial and coefficient given one of the three aforementioned monomial orderings.
\end{itemize}

\subsection{Definition of the GRUR}

We consider $f_1,\dots,f_n \in \C[W,X]$ a family of polynomials that generate a generically zero-dimensional ideal $\mathcal{I} := \langle f_1,\dots, f_n \rangle \subseteq \C[W,X]$. Then there exists a Zariski-closed subset $\mathcal{D}$ of $\C^s$ with measure 0 such that for all values of parameters $w \in \C^s \backslash \mathcal{D}$, the specialized ideal $\mathcal{I} (w) := \langle f_1(w,X), \dots, f_n(w,X) \rangle \subseteq \C[X]$ is zero-dimensional, or equivalently, the variety $\mathcal{V}_w := V(I(w)) \subseteq \C^s$ has dimension zero.

We want to extend the notion of RUR to parametric systems in the sense that we want to find a family of univariate polynomials $\mathcal{F} = \{h_0, h_1, h_{X_1}, \dots, h_{X_n}\} \subseteq \C(W)[U_0]$ such that for almost all values of parameters $w \in \C^s$, we have the equivalence  
\begin{equation}\label{equiv systemes}
\begin{cases} f_1(w, X) = 0 \\ ~~~~~~~~ \vdots \\ f_n (w, X) = 0 \end {cases} \iff \begin{cases}
	h_0 (U_0, a_1, \dots, a_n ; w) = 0 \\
	X_i = \frac{\frac{\partial h}{\partial U_i} (U_0, a_1, \dots, a_n ; w)}{\frac{\partial h}{\partial U_0} (U_0, a_1, \dots, a_n ; w)}, ~ 1 \leqslant i \leqslant n
\end{cases}.\end{equation}
with a preservation of the multiplicities. In other words, we want the family $\mathcal{F}$ to specialize to a (classical) RUR of $V(\mathcal{I}(w))$ in the sense of \cite[Definition 3.2]{Rouillier1998} for almost all values of parameters. That motivates the following definition.

\begin{definition}\label{def GRUR}
	A family $\mathcal{F} := \{h_0, h_1, h_{X_1}, \dots, h_{X_n}\}$ is called a \textit{Generic RUR} (GRUR) of $V(\mathcal{I}_W)$ (respectively a Generic RUR-candidate) if for almost all $w \in \C^s$ the family $\phi_w(\mathcal{F})$ defines a RUR (respectively a RUR-candidate) of $V(\mathcal{I}(w))$.
\end{definition}

We now see $f_1, \dots, f_n$ as elements of $\C(W)[X]$. They generate an ideal $\mathcal{I}_W \subseteq \C(W)[X]$. Denote by $\mathcal{G} \subseteq \C[W,X]$ a reduced Gröbner basis of $\mathcal{I}$ for the product ordering $<_{W,X}$. Then $\mathcal{G} \subseteq \C(W)[X]$ is also a Gröbner basis of $\mathcal{I}_W$ for the ordering $<_X$ (see for example \cite{LSeD2022}). However, it might not be reduced. We denote by $\mathcal{W}_{\mathcal{G}}$ the variety $\cup_{g \in \mathcal{G}} V(\LC_{<X}(g))$.

\begin{lemma}
    Denote by $\overline{\C(W)}$ an algebraic closure of $\C(W)$. Then the variety 
    $$V(I_W):=\big\{x \in \overline{\C(W)}^n \mid f_1(W,x)=\dots=f_n (W,x)=0\big\}$$
     has dimension zero.
\end{lemma}

\begin{proof}
    For all values of parameters $w$ outside of $\mathcal{D} \cup \mathcal{W}_{\mathcal{G}}$, the specialized ideal $\mathcal{I}(w)$ has dimension zero and $\mathcal{G}(w) := \{g(w,X), ~ g \in \mathcal{G}\}$ is a Gröbner basis of $\mathcal{I}(w)$ for $<_X$ (see \ref{spe BG}). 

    By the finiteness theorem (see \cite[§5.3 Theorem 6]{CLO2005}) there exists non-negative integers $m_1,\dots, m_n$ and $g_1, \dots, g_n \in \mathcal{G}$ such that each $g_i$ has leading monomial $\LM_{<X}(g_i) = X_i^{m_i}$. By hypothesis, we have that $\LM_{<X} (g_i (w,X)) = \LM_{<X} (g_i)$, which also means that $V(\mathcal{I}_W)$ has dimension zero. \qed
\end{proof}

Knowing that $V(\mathcal{I}_W)$ has dimension zero, we can follow the construction of the classical RUR as in the introduction. Consider a linear form $t:=a_1 X_1 + \dots + a_n X_n \in \C[X]$ (it does not involve $W$) and consider the polynomial $h := \prod_{x \in V(\mathcal{I}_W)} (U_0 - U_1 x_1 - \dots - U_n x_n)^{\mu(x)}$ where $\mu(x)$ is the multiplicity of $x$.

We then define polynomials $h_0, h_1, h_{X_1},\dots, h_{X_n}$ in the following way : 

\begin{equation}\label{formules F}
\begin{cases} h_0 (U_0 ; W) := h(U_0,a_1,\dots, a_n ; W) \\
    h_1(U_0 ; W) := \frac{\partial h}{\partial U_0} (U_0,a_1,\dots,a_n ; W) \\
    h_{X_i} (U_0 ; W) :=\frac{\partial h}{\partial U_i} (U_0,a_1,\dots,a_n ; W), ~ 1 \leqslant i \leqslant n \end{cases}.
 \end{equation}

\begin{definition}
	We say that the linear form $t \in \C[X]$ is \textit{generically separating} if $t$ separates $V(\mathcal{I}_W)$, that is, if $V(\mathcal{I}(w))$ for almost all values of $w \in \C^s$.
\end{definition}

There is a natural algebraic rephrasing of the geometric fact that $t$ separates $V(\mathcal{I}_W)$.

\begin{proposition}
	The linear form $t = a_1 X_1 + \dots + a_n X_n$ separates $V(\mathcal{I}_W)$ if and only if the discriminant $\Disc_{U_0} (\overline{h}(a_1, \dots, a_n ; W)) \in \C(W)$ is not identically zero where $\overline{\cdot}$ is the squarefree part with regard to the variable $U_0$.
\end{proposition}

\begin{proof}
	We recall that $h := \prod_{x \in V(\mathcal{I}_W)} (U_0 - x_1 U_1 - \dots - x_n U_n)^{\mu(x)}$. Then, the squarefree part with regard to $U_0$ is given by $\overline{h} = \prod_{x \in V(\mathcal{I}_W)} (U_0 - x_1 U_1 - \dots - x_n U_n)$. In other words, we get rid of the geometric multiplicities. When we substitute the values of $U_1, \dots, U_n$ to $a_1, \dots, a_n$, we obtain $\overline{h}(U_0, a_1, \dots, a_n ; W) = \prod_{x \in V(\mathcal{I}_W)} (U_0 - t(x))$. This polynomial is squarefree if and only if none of the linear factors appearing in the product repeat, that is, if and only if $t$ is injective on $V(\mathcal{I}_W)$.  \qed
\end{proof}

\begin{remark}
	A linear form that separates $V(\mathcal{I}_W)$ also separates $V(\mathcal{I}(w))$ for almost all $w \in \C^s$. The converse is also true, but two linear forms that both separate $V(\mathcal{I}_W)$ might differ through the algebraic sets of parameters for which they do not separate the specialized system.
\end{remark}

\begin{theorem}\label{theoreme GRUR}
	Consider $f_1, \dots, f_n \in \C[W,X]$ that generate a generically zero-dimensional ideal $\mathcal{I}$. Let $t = a_1 X_1 + \dots + a_n X_n \in \C[X]$ be a linear form. Then the family $\mathcal{F} := \{h_0, h_1, h_{X_1}, \dots, h_{X_n}\}$ defined by \ref{formules F} is a GRUR-candidate of $V(\mathcal{I}_W)$. 
	Moreover, if $t$ is generically separating, $\mathcal{F}$ is a GRUR of $V(\mathcal{I}_W)$.
\end{theorem}

The next subsections are there to give some ingredients to prove Theorem \ref{theoreme GRUR}. We specify in Corollary \ref{valeurs generiques GRUR} the parameters values to avoid.

\subsection{Van der Waerden-like approach}

They key-point in this approach is an idea already used by Renegar in \cite{Renegar1992}. The idea is that the polynomial $h$ (from which we can derive the elements of $\mathcal{F}$) is a factor of Van der Waerden's $u$-resultant. This resultant has great specialization properties we can use to study the specialization of $h$. Let us quickly recall Van der Waerden's idea (\cite[§11.83]{VdW1931}). We add a new equation given by $f_0=0$ where $f_0:=U_0 - U_1 X_1 - \dots - U_n X_n$. Here $U := [U_0, \dots, U_n]$ are independant new variables. We then homogeneize $f_0, \dots, f_n$ into $F_0, \dots, F_n$ using a new variable named $X_0$. Observe that $F_0 \in \Z[U][X_0, \dots, X_n]$ and $F_1, \dots, F_n \in \C[W][X_0, \dots, X_n]$. Now $F_0, \dots, F_n$ are $n+1$ polynomials in $n+1$ variables so their resultant is well defined. It is called the \textit{$u$-resultant} and we write it $\Res(F_0, \dots, F_n)$. It is an element of $\C[W][U_0, \dots, U_n]$.

\begin{remark}
    The $u$-resultant can be identically zero. In that case, the method that we describe below will fail. This is the case when the variety has a component of positive dimension for example. With our assumptions, $V(\mathcal{I}_W)$ is zero-dimensional so this will not happen. It can also be identically zero if the projective closure of $V(\mathcal{I}_W)$ has a point at infinity. This problem can be avoided by a generic change of variables (see \cite{Morgan1986}). 
\end{remark}

For the rest of this section, we suppose that $V(\mathcal{I}_W)$ has no point at infinity. Denote by $\tilde{F}_i := F_i (0, X_1,\dots, X_n)$. Under these assumptions, the resultant $\Res(\tilde{F_1}, \dots, \tilde{F}_n) \in \C[W]$ is not identically zero and we get the following result.

\begin{proposition}\label{facto u-res}
    The $u$-resultant factorizes as 
    \begin{equation}\label{factorisation u-res} \Res (F_0, \dots, F_n) (U ; W) = \Res(\tilde{F_1}, \dots, \tilde{F}_n) (W) \prod_{x \in V(\mathcal{I}_W)} (U_0 - x_1 U_1 - \dots - x_n U_n)^{\mu(x)}.
\end{equation}
\end{proposition}

\begin{proof}
    First, the Poisson formula for homogeneous resultants (see \cite[Proposition 2.7]{Jouanolou1-1991}) gives that $\Res(F_0, \dots, F_n) = \Res(\tilde{F_1}, \dots, \tilde{F_n}) \cdot \det(m_{f_0})$ where $m_{f_0}$ is the endomorphism of multiplication by $f_0$ in the quotient $\C(W)[X]/\mathcal{I}_W$.  
    We can then apply Stickelberger's theorem \ref{Stickelberger} to get the announced expression. \qed
\end{proof}

We recall that $h(U ; W) = \prod_{x \in V(\mathcal{I}_W)} (U_0 - x_1 U_1 - \dots - x_n U_n)^{\mu(x)}$ and that given $t = a_1 X_1 + \dots + a_n X_n$ we have $h_0 (U_0 ; W) = h(U_0, a_1, \dots, a_n ; W)$.

\subsection{Proof of Theorem \ref{theoreme GRUR}}

First we prove that the family $\mathcal{F} = \{h_0, h_1, h_{X_1}, \dots, h_{X_n}\}$ is a  GRUR-candidate. We show that $\phi_w(\mathcal{F})$ is a RUR-candidate of $V(\mathcal{I}(w))$ in the sense of \cite[Definition 3.2]{Rouillier1998}. Let us introduce a short lemma.

\begin{lemma}\label{egal fractions}
	For all $i \in [\![1,n]\!]$, 
	\begin{equation}\label{egalite fractions}
	\frac{\frac{\partial \overline{h}}{\partial U_i} (U_0,a_1,\dots,a_n ; W)}{\frac{\partial \overline{h}}{\partial U_0} (U_0,a_1,\dots,a_n ; W)} = \frac{h_{X_i} (U_0 ; W)}{h_1 (U_0 ; W)}\end{equation}
	where $\overline{\cdot}$ is the squarefree part with regard to $U_0$.
\end{lemma}

\begin{proof} 
    The left-hand side in Equation \eqref{egalite fractions} is just the reduced version of the right-hand side. By multiplying both the numerator and the denominator of the right-hand side by $\prod_{y \in V(\mathcal{I}_W)} (U_0 - x_1U_1 - \dots - x_n U_n)^{\mu(x) -1}$ we get the result. \qed
\end{proof}

The meaning of this lemma is that in the definition of the RUR given by Rouillier, using $h_0$ or its squarefree part to define the rational fractions does not change anything. In our case, we will use the formulas with multiplicities as they have better specialization properties. Thus, using that $\phi_w$ commutes with derivatives, the only remaining thing to show is that for almost all $w \in \C^s$, $h_0(U_0 ; w)$ is the characteristic polynomial of the multiplication by  $t$ in the quotient $\C[X]/\mathcal{I}(w)$, that is, that $h_0 (U_0 ; w) = \prod_{x \in V(\mathcal{I}(w))} (U_0 - t(x))^{\mu(x)}$. 

With that objective in mind, let us summarize the algebraic sets of $\C^s$ that will appear in what follows :
\begin{itemize}
    \item we recall that $\mathcal{D}$ is the set of parameters for which the specialized system is not zero-dimensional ;
    \item we denote by $\mathcal{L}:= \cup_{i=0}^n V(\LT_{<X}(f_i)) \subseteq \C^s$ the set of parameters for which one or more of the leading coefficient in $X$ of the initial equations vanish ;
    \item we denote by $\mathcal{S}$ the set of parameters $w$ for which $t$ is not injective on $V(\mathcal{I}(w))$ ;
    \item we denote by $\mathcal{W}_{\infty}$ the set of parameters for which the specialized system has roots at infinity. 
\end{itemize}

\begin{remark}\label{def W_infty}
    The fact that $\mathcal{W}_\infty$ is an algebraic set comes from the fact that $\mathcal{W}_\infty = V( \Res(\tilde{F_1}, \dots, \tilde{F}_n))$.
\end{remark}

We start by studying the specialization of formula \eqref{factorisation u-res}.

\begin{lemma}
    For any value of parameters $w \in \C^s \backslash \mathcal{L}$, both resultants appearing in \ref{factorisation u-res} specialize. In other words,
\begin{equation}\label{specialization uRes}
    \phi_w \Big ( \Res (F_0, \dots, F_n) \Big ) = \Res \big( \phi_w(F_0), \dots, \phi_w(F_n)\big)
\end{equation}

and 

\begin{equation}\label{specialization Res}
        \phi_w \left( \Res (\tilde{F_1}, \dots, \tilde{F_n})\right) = \Res \left( \tilde{\phi_w(F_1)}, \dots, \tilde{\phi_w(F_n)}\right)
\end{equation}
\end{lemma}

\begin{proof}
Resultants specialize as long as the degrees of their arguments do not change. Outside of $\mathcal{L}$, the degrees of $f_0, \dots, f_n$ and thus those of $F_0, \dots, F_n$ and $\tilde{F_1},\dots, \tilde{F}_n$ do not change after specializing so both resultants specialize. \qed
\end{proof}

That specialization property makes us able to identify what the polynomial $h(U ; W)$ specializes to. 

\begin{lemma}\label{lemme spe Chow form}
    For any value of parameters $w \in \C^s \backslash (\mathcal{D} \cup \mathcal{L} \cup \mathcal{W}_\infty)$, we have 
\begin{equation}\label{spe Chow form}
        \phi_w \left(h(U ; W) \right) = \prod_{x \in \mathcal{V}_w} (U_0 - x_1 U_1 - \dots - x_n U_n)^{\mu(x)}.
    \end{equation}
\end{lemma}

\begin{proof}
    First, we can use one more time the Poisson formula from \cite[Proposition 2.7]{Jouanolou1-1991} and Stickelberger's theorem to get that 
    $$\Res \left( \phi_w (F_0), \dots, \phi_w (F_n)\right) =$$
    \begin{equation}\label{Poisson spe}
        \Res(\tilde{\phi_w(F_1)}, \dots, \tilde{\phi_w(F_n)}) \prod_{x \in V(\mathcal{I}(w))} (U_0 - x_1 U_1 - \dots - x_n U_n)^{\mu(x)}.
    \end{equation}
    We can then substitute \eqref{specialization uRes} and \eqref{specialization Res} into \eqref{Poisson spe}. Because $w \notin \mathcal{W}_\infty$, the specialized system has no root at infinity, and thus formula \eqref{Poisson spe} is non trivial. We can compare it to the specilization of \eqref{factorisation u-res} to yield the wanted formula. \qed
\end{proof}

Now we just have to substitute the values $U_1, \dots, U_n$ to $a_1, \dots, a_n$ in \eqref{spe Chow form} to get 
\begin{equation}
     \phi_w \left(h(U_0, a_1, \dots, a_n ; W) \right) = \phi_w (h_0) = \prod_{x \in V(\mathcal{I}(w))} (U_0 - t(x))^{\mu(x)}.
\end{equation}
and the right hand side of this equation is exactly the characteristic polynomial of the multiplication by $t$ in the quotient $\C[X]/\mathcal{I}(w)$ (this is once again Stickelberger's theorem). We summarize everything in the following proposition.

\begin{proposition}\label{GRUR-candidate}
	For any values of the parameters $w \in \C^s \backslash (\mathcal{D} \cup \mathcal{L} \cup \mathcal{W}_\infty)$ the family $\phi_w(\mathcal{F})$ is a RUR-candidate of $V(\mathcal{I}(w))$.
\end{proposition}

We now suppose that the linear form $t$ is generically separating and we denote by $\mathcal{S} \subseteq \C^s$ the algebraic set of measure 0 formed by the $w \in \C^s$ for which $t$ is not injective on $V(\mathcal{I}(w))$. We have an immediate corollary.

\begin{corollary}\label{valeurs generiques GRUR}
	For any values of the parameters $w \in \C^s \backslash (\mathcal{D} \cup \mathcal{L} \cup \mathcal{W}_\infty \cup \mathcal{S})$, the family $\phi_w(\mathcal{F})$ is a RUR  of $V(\mathcal{I}(w))$.
\end{corollary}

\begin{proof}
	We know from Proposition \ref{GRUR-candidate} that $\mathcal{F}$ specializes to a RUR-candidate of $V(\mathcal{I}(w))$. It becomes a RUR whenever $t$ separates $V(\mathcal{I}(w))$. By definition, this is the case when $w \notin \mathcal{S}$.
\end{proof}

\begin{remark}\label{generic sets}
	We have the inclusion $\mathcal{D} \subseteq \mathcal{W}_{\infty}$ so the generic condition can be expressed as $w \in \C^s \backslash (\mathcal{W}_\infty \cup \mathcal{L} \cup \mathcal{S})$. The statement comes from the fact that if $w \in \mathcal{D}$, then $V(\mathcal{I}(w))$ has dimension zero and thus its projective closure intersects the hyperplane at infinity, that is, $w \in \mathcal{W}_\infty$.
\end{remark}

\section{Size of the GRUR}

We present in this section some bounds that can be obtained on the degrees and size of the polynomials of the GRUR.

In what follows, we will use the notation $\mathfrak{h} (f) := \log_2(\|f\|_\infty)$ for the height of a polynomial $f$ with integer coefficients. We suppose for this section that $f_1,\dots,f_n$ have integer coefficients by chasing denominators. We denote by :
\begin{itemize}
    \item $d_{1,X},\dots,d_{n,X}$ the total degrees of $f_1,\dots,f_n$ in $X$ and $\displaystyle{d_X:=\max_{1\leqslant i \leqslant n} d_{i,X}}$ ;
    \item $d_{1,W},\dots,d_{n,W}$ the total degrees of $f_1,\dots,f_n$ in $W$ and $\displaystyle{d_W:=\max_{1\leqslant i \leqslant n} d_{i,W}}$ ;
    \item $c:=\displaystyle{\max_{1 \leqslant i \leqslant n} \|f_i\|_{\infty}}$ the maximum coefficient among $f_1,\dots, f_n$ and $\tau := \displaystyle{\max_{1 \leqslant i \leqslant n}} \mathfrak{h} (f_i)$.
\end{itemize}

\subsection{Bounds for $h$}

The work presented in this section makes great used of the arithmetic Nullstellensätz presented in \cite{DKS2012}. We warmly thank Teresa Krick for the reference.

\begin{theorem}\label{Bornes h(U)}
    There exists $\Delta \in \Z[W][U] $ and $P \in \Z[W]$ such that 
    $$h(U ; W) = \frac{\Delta(U ; W)}{P(W)}$$
    with the following bounds :
    \begin{itemize}
        \item $\deg_U(\Delta) \leqslant d_X^n$ ;
        \item $\deg_W (\Delta) \leqslant (n+1) d_X^n d_W$ ;
        \item $\deg_W (P) \leqslant nd_X^{n-1} d_W$ ;
        \item $\mathfrak{h}(\Delta), \mathfrak{h}(P) \leqslant O \big(nd_X^n \tau + nd_X^n d_W \log_2(s+1) \big)$.
    \end{itemize}
\end{theorem}

\begin{proof}
We start by following the ideas of \cite[Example 4.18 and 4.37]{DKS2012} and consider the generic polynomials of total degrees $d_0, \dots, d_n$
$$\mathbf{F}_j := \sum_{|\alpha| = d_j} u_{j, \alpha} X^{\alpha}, ~ 0 \leqslant j \leqslant n$$
where $\mathbf{u}_j := \{u_{j, \alpha}, ~ |\alpha| = d_j\}$ and $X^\alpha = X_0^{\alpha_0} \cdots X_n^{\alpha_n}$. Denote by $\mathbf{u} := \{\mathbf{u}_1, \dots, \mathbf{u}_n\}$. These generic coefficients will later be specialized to those of $F_0, \dots, F_n$ of the previous section. 

The arithmetic Nullstellensätz in \cite[Theorem 4.28]{DKS2012} applied to $V = \mathbb{A}^n_{\Q}$ and $\mathbf{F}_0, \dots, \mathbf{F}_n$ states that there exists $\alpha \in \Z[\mathbf{u}]$ and $g_0, \dots, g_n \in \Z[\mathbf{u}, X]$ such that 
\begin{equation}\label{DKS}
    g_0 (\mathbf{u}, X) \mathbf{F}_0 + \dots + g_n (\mathbf{u}, X) \mathbf{F}_n = \alpha
\end{equation}
with the following bounds on the degrees and the height:
\begin{itemize}
    \item $\deg_{\mathbf{u}_l} (\alpha) \leqslant \displaystyle{\prod_{j=0 \atop j \neq l}^n d_j}, ~ 0 \leqslant l \leqslant n$ ;
    \item $\mathfrak{h}(\alpha) \leqslant d_X^{n-1} (6n + 10)\log_2 (n+2)$.
\end{itemize}

The resultant $\Res(\mathbf{F}_0, \dots, \mathbf{F}_n)$ generates the elimination ideal $\langle \mathbf{F}_0, \dots, \mathbf{F}_n \rangle \cap \Z[\mathbf{u}]$ so there exists some $\Lambda \in \Z[\mathbf{u}]$ such that $\alpha = \Lambda \Res(\mathbf{F}_0, \dots, \mathbf{F}_n)$. It is well known (see for example \cite[§3.3, Theorem 3.1]{CLO2005}) that $\Res(\mathbf{F_0}, \dots, \mathbf{F_n})$ has partial degree $\deg_{u_l}$ equal to $\prod_{j=0 \atop j \neq l}^n d_j$. In other words, $\Lambda$ has partial degree in each variable equal to zero, that is, $\Lambda \in \Z$.

We can then specialize the coefficients $\mathbf{u}$ to those of $F_0, \dots, F_n$. Relation \eqref{DKS} transforms into
\begin{equation*}
 G_0(W, X) F_0 + \dots + G_n(W, X) F_n = \Lambda \Res(F_0, \dots, F_n)
\end{equation*}
where $G_0, \dots, G_n \in \Z[W, X]$. 

We know from Proposition \ref{facto u-res} that the $\Res(F_0, \dots, F_n)$ factorizes into a product $C(W) h(U ; W)$ where $C = \Res(\tilde{F}_1, \dots, \tilde{F}_n)$.

 Let us denote by $\Delta := \Lambda \Res(F_0, \dots, F_n) \in \Z[W][U]$. That means we have
\begin{equation}\label{quotient R}
    h(U ; W) = \frac{\Delta(U ; W)}{P(W)}
\end{equation}
where $P(W) := \Lambda C(W)$.

What we want now is a bound on the degree in $W$ and a bound on the height of $\Delta$ and $\Lambda C$ to derive bounds for the polynomials of the GRUR. Thanks to the definition of $h(U ; W)$, we know that $\deg_U(\Delta) \leqslant d_X^n$.

The bounds given by \cite[Theorem 4.28]{DKS2012} on Equation \eqref{DKS} can be derived into bounds for $\Delta$. Indeed, we specialize Equation \eqref{DKS} with the coefficients of $F_0, \dots, F_n$ which have degrees in $W$ at most $d_W$. For each $l \in [\![0, n]\!]$, $\alpha$ has partial degree in $\mathbf{u}_l$ bounded by $d_X^n$ so the total degree in $\mathbf{u}$ is bounded by $(n+1)d_X^n$ and thus $\deg_W(\Delta) \leqslant (n+1)d_X^n d_W$.

For the bound on the height, we use a lemma presented in \cite[Lemma 2.37]{DKS2012} 
\begin{lemma}\label{lemme hauteurs DKS}
    Let $f_1, \dots, f_r \in \Z[x_1, \dots, x_k]$ with total degrees bounded by $d$, bitsize bounded by $\tau$, and let $g \in \Z[y_1, \dots, y_r]$. Then
    $$\mathfrak{h}(f_1 \cdots f_r) \leqslant \sum_{i=1}^r \mathfrak{h}(f_i) + \log_2(k+1) \sum_{i=2}^r \deg(f_i),$$
    and
    $$\mathfrak{h}(g(f_1, \dots, f_r)) \leqslant \mathfrak{h}(g) + \deg(g) (\tau + \log_2(r+1) + d\log_2(k+1)).$$ 
\end{lemma}

In our case, the $f_i$'s are the coefficients of $F_0, \dots, F_n$ so $r \leqslant n(d_X+1)^n + (n+1)$ and $k = s$. Using the bounds on $\alpha$, and keeping only the dominant terms, we get
$$\mathfrak{h}(\Delta) \leqslant O \big(nd_X^n \tau + nd_X^n d_W \log_2(s+1) \big).$$

In addition to that, since $\mathfrak{h}(\Lambda) \leqslant \mathfrak{h}(\Delta)$, that bound is also true for $\Lambda$.

We can follow a very similar reasoning to get an upper bound on the height of $C$ as it is the resultant of $\tilde{F_1}, \dots, \tilde{F}_n$. Indeed, $\tilde{F}_1, \dots, \tilde{F}_n$ have the same degrees in $X$ as $F_1, \dots, F_n$ and we can take the bound $\tau$ for their height. Denote by $\mathbf{u}':= \mathbf{u}_1, \dots, \mathbf{u}_n$. We first apply theorem 4.28 from \cite{DKS2012} to the generic polynomials $\mathbf{F}_1, \dots, \mathbf{F}_n$ and thus get the existence of $\beta \in \Z[\mathbf{u}']$ and $z_1, \dots, z_n \in \Z[\mathbf{u}', X]$ such that 
$$z_1(\mathbf{u}', X) \mathbf{F}_1 + \dots + z_n(\mathbf{u}', X) \mathbf{F}_n = \beta$$ 
with the following bounds on the degrees and the height:
\begin{itemize}
    \item $\deg_{\mathbf{u}_l} (\beta) \leqslant \displaystyle{\prod_{j=1 \atop j \neq l}^n d_j} \leqslant d_X^{n-1}, ~ 1 \leqslant l \leqslant n$ ;
    \item $\mathfrak{h}(\beta) \leqslant d_X^{n-2} (6n + 10)\log_2 (n+2)$.
\end{itemize}

In addition to that, $\beta = \Lambda' \Res(\mathbf{F}_1, \dots, \mathbf{F}_n)$ with $\Lambda' \in \Z$ just like before. We can then use lemma \ref{lemme hauteurs DKS} to yield a bound on the height of $\Lambda' C$ which can be kept as a bound on the height of $C$. Once again we only keep the dominant terms.
\begin{equation}\label{hauteur C}
    \mathfrak{h}(C) \leqslant O \big( nd_X^{n-1} \tau + n d_X^{n-1} d_W \log_2(s+1)\big).
\end{equation}

Since $\mathfrak{h}(\Lambda C) = \mathfrak{h}(\Lambda) + \mathfrak{h}(C)$, we get $\mathfrak{h}(P) = \mathfrak{h}(\Lambda C) \leqslant O \big(nd_X^n \tau + nd_X^n d_W \log_2(s+1) \big)$.

For the degree, we use the fact that $\Res(\mathbf{F}_1, \dots, \mathbf{F}_n)$ has total degree in $\mathbf{u}'$ bounded by $n d_X^{n-1}$, and thus that $C$ (which means also $P$) has degree in $W$ at most $n d_X^{n-1}d_W$. \qed
\end{proof}

\subsection{Bounds for the GRUR polynomials}

Now that we have a bound on the degree and the height of $h(U ; W)$, we can deduce bounds on the degree and the height of the polynomials of the family $\mathcal{F}$. We write that $f = \tilde{O}(g)$ when $f = O(g \log_2^k(g))$ for some $k \geqslant 0$.

\begin{corollary}\label{bounds RUR}
    Suppose that $ = a_1 X_1 + \dots + a_n X_n \in \Q[X]$ is a generically separating form. Then the polynomials $h_0, h_1, h_{X_1}, \dots, h_{X_n}$ of the GRUR have a height bounded by
    $$O \big(nd_X^n \tau + nd_X^n d_W \log_2(s+1) + n^2 d_X^n \log_2(d_X) \big) = \tilde{O} \big(nd_X^n \tau + nd_X^n d_W \log_2(s+1)\big)$$
    and degree in $W$ bounded by 
    $$ \kappa := (n+1)d_X^n d_W$$
     We consider the height and the degree to be the maximum of those of the respective numerators and the denominators.
\end{corollary}

\begin{proof}
    Denote by $\tau_t$ the height of the linear form $t$. In \cite[Lemma 2.1]{Rouillier1998}, it is shown that one can find a linear form that separates $V(\mathcal{I}_W)$ in the set $\mathcal{T} = \big\{X_1 + kX_2 + \dots + k^{n-1} X_{n}, ~ 0 \leqslant k \leqslant nD(D-1)/2\big\}$. In other words, the chosen linear form will have its coefficients not bigger than $(nd_X^{2n})^{n-1}$ which means that $\tau_t \leqslant O(n\log_2(n) + n^2 \log_2(d_X))$.
    
    Let us begin with the polynomial $h_0 (U_0 ; W)$ which is obtained as $h_0 (U_0 ; W) = h(U_0, a_1, \dots, a_n ; W) = \Delta(U_0, a_1, \dots, a_n ; W)/P(W)$. Specializing the variables $U$ does not affect the denominator $P$. It does not change the bound on the degree in $W$ of $\Delta$ either. For the height, we can use once more lemma \ref{lemme hauteurs DKS} to get 
 \begin{eqnarray*}
 \mathfrak{h}(\Delta(U_0, a_1, \dots, a_n ; W)) & \leqslant & \mathfrak{h}(\Delta) + \deg_U(\Delta)(\tau_t + \log_2(n+1)) \\
 & \leqslant & \mathfrak{h}(\Delta) + d_X^n \tau_t + d_X^n \log_2(n+1)
 \end{eqnarray*}
    and by keeping the dominant terms, 
    \begin{equation*}
        \mathfrak{h}(\Delta(U_0, a_1, \dots, a_n ; W)) \leqslant O(nd_X^n \tau + nd_X^n d_W \log_2(s+1) + n^2 d_X^n \log_2(d_X)). 
    \end{equation*}

    Then the polynomials $h_1, h_{X_1}, \dots, h_{X_n}$ are obtained as 
    $$h_1 (U_0 ; W) = \frac{\partial h}{\partial U_1} (U_0, a_1, \dots, a_n ; W) = \frac{\frac{\partial \Delta}{\partial U_1} (U_0, a_1, \dots, a_n ; W)}{P(W)}$$
    and
    $$ h_{X_k} (U_0 ; W) = \frac{\partial h}{\partial U_k} (U_0, a_1, \dots, a_n ; W) = \frac{\frac{\partial \Delta}{\partial U_k} (U_0, a_1, \dots, a_n ; W)}{P(W)}, ~ 1 \leqslant k \leqslant n.$$
    
    Each $\partial \Delta/\partial U_i$ has height bounded by $\mathfrak{h}(\Delta) + n\log_2(d_X)$ and using once more lemma \ref{lemme hauteurs DKS} we get
    \begin{eqnarray*}
        \mathfrak{h} \left(\frac{\partial \Delta}{\partial U_i} (U_0, a_1, \dots, a_n ; W) \right) & \leqslant & \mathfrak{h}(\Delta) + n\log_2(d_X) + d_X^n \tau_t + d_X^n \log_2(n+1) \\
        & \leqslant & O(nd_X^n \tau + nd_X^n d_W \log_2(s+1) + n^2 d_X^n \log_2(d_X)).
    \end{eqnarray*}
Since the height of $P$ is not bigger than that, the bounds on the height of $h_1, h_{X_1}, \dots, h_{X_n}$ are proved.
\end{proof}

We present in the two following sections two Las-Vegas algorithm to compute the polynomials of $\mathcal{F}$. The first one relies on linear algebra and the second one is based on a evaluation/interpolation scheme.

\section{First algorithm : linear algebra approach}

    In this subsection we will simultaneously work with the generic quotient $\mathcal{A}_W := \C(W)[X]/\mathcal{I}_W$ and the specialized quotient for values  $w \in \C^s$ that we denote by $\mathcal{A}(w) :=  \C[X]/\mathcal{I}(w)$.
    Every object related to one of these quotients will have an exponent $\mathcal{A}_W$ or $\mathcal{A}(w)$ (or an index) to make the distinction. This distinction will be particularly useful to show specialization properties. For any $v \in \C(W)[X]$, we will denote by respectively $\Trace^{\mathcal{A}_W}(v)$ and $\chi^{\mathcal{A}_W}_v$ the trace and characteristic polynomial of 
    $$\fonc{m^{\mathcal{A}_W}_v}{\mathcal{A}_W}{\mathcal{A}_W.}{[f]_{\mathcal{A}_W}}{[fv]_{\mathcal{A}_W}}$$
    We choose analogous notations for $v \in \C[X]$ by replacing $\mathcal{A}_W$ by $\mathcal{A}(w)$ in the previous notations when we specialize to $w \in \C^s\backslash \mathcal{D}$.

\subsection{Pre-computations}

This approach to the computation of the polynomials in the family $\mathcal{F}$ takes some pre-requisites that we mention now. 

$\cdot$ \underline{Gröbner basis :} we denote by $\mathcal{G}$ a reduced Gröbner basis of $\mathcal{I} \subseteq \C[W,X]$ for the block ordering $<_{W,X}$. 
We recall that $\mathcal{G}$ when seen as a subset of $\C(W)[X]$ can also be interpreted as a Gröbner basis of $\mathcal{I}_W$ for the ordering $<_X$. In what follows, we consider that $\mathcal{G}$ has been reduced.

$\cdot$ \underline{Basis of the quotient :} a basis $\mathcal{B}_{\mathcal{A}_W} := \{[p_1]_{\mathcal{A}_W}, \dots, [p_D]_{\mathcal{A}_W}\}$ of $\mathcal{A}_W$ can be computed from $\mathcal{G}$ by keeping all the monomials not reducible by $\mathcal{G}$. 

$\cdot$ \underline{Multiplication matrices :} we also need the multiplication matrices by each variables $M^{\mathcal{A}_W}_{X_1}, \dots, M^{\mathcal{A}_W}_{X_n}$ in the quotient $\mathcal{A}_W$. This can be done by following FGLM's approach as in \cite[Proposition 2.1]{FGLM1993}. 

$\cdot$ \underline{Multiplicative tensor :} to avoid repeating the same computations over and over, one can compute the multiplicative tensor $$\MT (\mathcal{A}_W) := \{\overrightarrow{[p_ip_j]}_{\mathcal{A}_W}, ~ 1 \leqslant i, j \leqslant D\}$$ where $\overrightarrow{[\cdot]}_{\mathcal{A}_W}$ are the coordinates with regard to the basis $\mathcal{B}_{\mathcal{A}_W}$.

\subsection{Computation of $\mathcal{F}$}

The formulas that we present in this subsection are directly adapted from Rouillier's work in \cite[§4]{Rouillier1998}. 

\subsubsection{Computation of $h_0$ and $h_1$}

We recall that $h_0 = \chi_t^{\mathcal{A}_W}$. That being said, Newton's formulas translate really nicely in terms of traces as shown in the following lemma.

\begin{lemma}\label{relations Newton}
    Let us write $\chi_t^{\mathcal{A}_W} = \sum_{i=0}^D b_i U_0^{D-i}$ with $b_0 = 1$. Then the coefficients of $\chi_t^{\mathcal{A}_W}$ satisfy the following equations :
    \begin{equation}\label{identites Newton}
        (D-k)b_k = \sum_{i=0}^k b_{k-i} \Trace^{\mathcal{A}_W}(t^i), ~~~~0 \leqslant k \leqslant D .
    \end{equation}
\end{lemma}

The equations given by \eqref{identites Newton} form a triangular linear system, which is convenient to solve. In addition to that, the system only depends on the traces $\Trace^{\mathcal{A}_W} (t^i)$.

Once $\chi_t^{\mathcal{A}_W} = h_0$ is obtained, we directly get $h_1$ by taking its derivative.

\subsubsection{Computation of $h_{X_1}, \dots, h_{X_n}$}

We can express the polynomials $h_{X_1}, \dots, h_{X_n}$ using computable closed formulas depending only on $t$ and $h_0$. They are similar to \cite[Theorem 3.1]{Rouillier1998} formulas when we decide to keep the fractions of the GRUR unreduced (see Lemma \ref{egal fractions}).

\begin{lemma}\label{formules h_Xi}
    Sticking with the notation $h_0 = \chi_t^{\mathcal{A}_W} = \sum_{i=0}^D b_i T^{D-i}$, we get that for each $j \in [\![1,n]\!]$, 
    \begin{equation}
        h_{X_j} (U_0 ; W) = \sum_{i=0}^{D-1} \Trace^{\mathcal{A}_W} (X_j t^i) H_{D-i-1} (U_0 ; W)
    \end{equation}
    where $H_{i} (U_0 ; W) := \sum_{j=0}^i b_j U_0^{i-j}$ is the $i$-th Hörner polynomial of $h_0$.
\end{lemma}

\begin{proof}
    Let us recall that $\chi_t^{\mathcal{A}_W} (U_0) = \prod_{x \in V(\mathcal{I}_W)} (U_0 - t(x))^{\mu(x)}$. We can then rewrite $h_{X_i}$ as
    \begin{eqnarray*}
        h_{X_i} (U_0 ; W) & = & \sum_{x \in V(\mathcal{I}_W)} \mu(x) x_i (U_0 - t(x))^{\mu(x) -1} \prod_{y \in V(\mathcal{I}_W) \atop y \neq x} (U_0 - t(y))^{\mu(y)} \\
        & = & \sum_{x \in V(\mathcal{I}_W)} \frac{\mu(x) x_i \prod_{y \in V(\mathcal{I}_W)}(U_0-t(y))^{\mu(y)}}{U_0 - t(x)} \\
        & = & \chi_t^{\mathcal{A}_W} (U_0) \sum_{x \in V(\mathcal{I}_W)} \frac{\mu(x) x_i}{U_0 - t(x)} \\
        & = & \chi_t^{\mathcal{A}_W} (U_0) \sum_{x \in V(\mathcal{I}_W)} \sum_{i=0}^{+\infty}  \frac{\mu(x) x_i t^i (x)}{U_0^{i+1}} = \chi_t^{\mathcal{A}_W} (U_0) \sum_{i=0}^{+\infty} \frac{\Trace^{\mathcal{A}_W}(X_i t^i)}{U_0^{i+1}}.
    \end{eqnarray*}
The last equality is obtained by permuting the sommation symbols and using the expression of $\Trace(X_i t^i)$ given by Stickelberger's theorem. We can then use the expanded expression of $\chi_t^{\mathcal{A}_W}$ to get that 

\begin{equation}\label{double somme}
    h_{X_i} (U_0 ; W) = \sum_{i=0}^{+\infty} \sum_{j=0}^D b_j \Trace^{\mathcal{A}_W}(X_i t^i) U_0^{D-j-i-1}.
\end{equation}

The left part of \eqref{double somme} is a polynomial so there must be bounds on the values of $i$ and $j$ in the above sums. More precisely, $D - i- j- 1$ must be non-negative. That means $i$ must be less or equal than $D-1$, and for each value of $i$, $j$ must be less or equal than $D-i-1$. Then \eqref{double somme} becomes
\begin{eqnarray*}
h_{X_i} (U_0 ; W) & = & \sum_{i=0}^{D-1} \Trace^{\mathcal{A}_W}(X_i t^i) \sum_{j=0}^{D-i-1} b_j U_0^{D-j-i-1} \\
			   & = & \sum_{i=0}^{D-1} \Trace^{\mathcal{A}_W}(X_i t^i) H_{D-i-1} (U_0 ; W)
\end{eqnarray*}
which is the announced formula. \qed
\end{proof}

\subsection{Specialization of the formulas}

We now justify the good specialization of the formulas given in Lemma \ref{formules h_Xi}. The main dependency comes from traces so we need conditions for them to specialize well. 

\subsubsection{Specialization of the Gröbner basis}

Let us consider $w \in \C^s$ and $\mathcal{G}(w) := \{\phi_w(g), ~ g \in \mathcal{G}\}$. We recall that $\mathcal{W}_{\mathcal{G}} = \cup_{g \in \mathcal{G}} V(\LC_{<X}(g))$. The following lemma presents a specialization property of the parametric Gröbner basis. In \cite[Theorem 2.6]{NY1999}, a similar result is presented for modular Gröbner bases.

\begin{lemma}\label{spe BG}
    For all $w \in \C^s\backslash \mathcal{W}_{\mathcal{G}}$, the family $\mathcal{G}(w)$ is a Gröbner basis of the specialized ideal $\mathcal{I}(w)$ for the ordering $\grevlex(X)$.
\end{lemma}

\begin{proof}
    All the reduction sequences of the $S$-polynomials of elements of $\mathcal{G}$ commute with the specialization morphism $\phi_w$ because the only denominators that appear are products of leading terms for $<_X$ of elements of $\mathcal{G}$ and by hypothesis they do not vanish.
    
    In addition to that, every $S$-polynomial reduces to zero because $\mathcal{G}$ is a gröbner basis. This remains true for the $S$ polynomials of elements of $\mathcal{G}(w)$ which are the specialization of $S$-polynomials of elements of $\mathcal{G}$.
    The conclusion is given by Buchberger's criterion. \qed
\end{proof}

This result has an important consequence on the dimension of the quotients $\mathcal{A}(w)$.

\begin{corollary}\label{common basis}
    For all $w \in \C^s \backslash (\mathcal{D} \cup \mathcal{W}_\mathcal{G})$, the quotients $\mathcal{A}_W$ and $\mathcal{A}(w)$ have the same dimension $D$ as vector spaces. 

    Moreover, if $\mathcal{B}_{\mathcal{A}_W} = \{[p_1]_{{\mathcal{A}_W}}, \dots, [p_D]_{\mathcal{A}_W}\}$ is a $\C(W)$-basis of $\mathcal{A}_W$ then $\mathcal{B}(w) := \{[p_1]_{\mathcal{A}(w)},\dots, [p_D]_{\mathcal{A}(w)}\}$ is a $\C$-basis of $\mathcal{A}(w)$.
\end{corollary}

\begin{proof}
    This comes from the fact that with our hypothesis, the elements of $\mathcal{G}$ and $\mathcal{G}(w)$ have the same leading monomials, which means that the non reducible monomials are the same for both Gröbner bases. \qed
\end{proof}

\subsubsection{Specialization of the normal forms}

Denote by $\NF_\mathcal{G} (f)$ the normal form of an element $f$ of $\C(W)[X]$ with regard to $\mathcal{G}$ and by $\NF_{\mathcal{G}(w)} (g)$ the normal form of an element $g$ of $\C[X]$ with regard to $\mathcal{G}(w)$. The key point for the specialization properties is the following proposition.

\begin{proposition}\label{specialization normal forms}
    Let $w \in \C^s \backslash (\mathcal{D} \cup \mathcal{W}_\mathcal{G})$ and $v \in \C(W)[X]$ with no denominator vanishing at $w$. Then we have 
    $$\NF_{\mathcal{G}(w)} \left( \phi_w (v)\right) = \phi_w \left( \NF_{\mathcal{G}} (v)\right).$$
\end{proposition}

\begin{proof}
    We start by saying that $v$ can be uniquely written as 
    \begin{equation}\label{decomp v}
        v = \sum_{g \in \mathcal{G}} v_g \cdot g + \NF_{\mathcal{G}} (v)
    \end{equation} 
    where each $v_g$ is an element of $\C(W)[X]$. They are obtained by successive reductions using $\mathcal{G}$ so the only new denominators potentially appearing are products of leading terms of $\mathcal{G}$. In other words, there is no denominator vanishing at $w$ in \eqref{decomp v} which means that we can specialize it at $w$ yielding 
    \begin{equation}\label{spe decomp v}
        \phi_w (v) = \sum_{g \in \mathcal{G}} \phi_w (v_g) \cdot \phi_w(g) + \phi_w(\NF_{\mathcal{G}} (v)). 
    \end{equation}
    By hypothesis we know that each $g \in \mathcal{G}$ satisfies the relation $\LM_{<X} (g) = \LM_{<X} (\phi_w(g))$. In addition to that, each monomial in $X$ appearing in $\NF_{\mathcal{G}} (v)$ (and consequently, in $\phi_w(\NF_{\mathcal{G}} (v))$ )is not divisible by any leading monomial of elements of $\mathcal{G}$ (and consequently by any leading monomial of elements of $\mathcal{G}(w)$).

    We know by lemma \ref{spe BG} that $\mathcal{G}(w)$ is a Gröbner basis of $\mathcal{I}(w)$ so the normal form of $\phi_w(v)$ with regard to $\mathcal{G}(w)$ is uniquely defined, and we can identify it to $\phi_w(\NF_{\mathcal{G}} (v))$ thanks to \eqref{spe decomp v}. \qed
\end{proof}

A direct consequence of this is the specialization of the coordinates when going from $\mathcal{A}_W$ to $\mathcal{A}(w)$.

\begin{corollary}\label{spe coordinates}
    Let $w \in \C^s \backslash (\mathcal{D} \cup \mathcal{W}_\mathcal{G})$ and $v \in \C(W)[X]$ with no denominator vanishing at $w$. Denote by $\overrightarrow{[\cdot]_{\mathcal{A}_W}}$ and $\overrightarrow{[\cdot]_{\mathcal{A}(w)}}$ the respective coordinates in $\mathcal{A}_W$ and $\mathcal{A}(w)$, respectively with regard to the bases $\mathcal{B}_{\mathcal{A}_W}$ and $\mathcal{B}_{\mathcal{A}(w)}$. Then we have 
    $$\overrightarrow{[\phi_w(v)]_{\mathcal{A}(w)}} = \phi_w \left(\overrightarrow{[v]_{\mathcal{A}_W}}\right)$$
    where $\phi_w$ is naturally applied to every coordinate.
\end{corollary}

\begin{proof}
    The starting point is the fact that any element of $\C(W)[X]$ or $\C[X]$ is represented in $\mathcal{A}_W$ or $\mathcal{A}(w)$ by its normal form. In particular, thanks to Corollary \ref{spe coordinates} we get that 
    $$[\phi_w(v)]_{\mathcal{A}(w)} = [\NF_{\mathcal{G}(w)} (\phi_w(v))]_{\mathcal{A}(w)} = [\phi_w(\NF_{\mathcal{G}} (v))]_{\mathcal{A}(w)}.$$

   Next, write $\NF_{\mathcal{G}} (v) = \sum_{i=1}^n c_i p_i$ where each $c_i$ is an element of $\C(W)$ with no vanishing denominator. That translates to $[\NF_{\mathcal{G}} (v)]_{\mathcal{A}_W} = (c_1, \dots, c_n)$. We can then apply $\phi_w$ to get that 
   \begin{equation}\label{pas d'idee}
    \phi_w(\NF_{\mathcal{G}} (v)) = \sum_{i=1}^n \phi_w(c_i) p_i.
   \end{equation} 
   We now recall (Corollary \ref{common basis}) that $\mathcal{A}_W$ and $\mathcal{A}(w)$ share the same monomials in their respective bases so we can interpret \ref{pas d'idee} as $$\overrightarrow{[\NF_{\mathcal{G}} (v)]_{\mathcal{A}_W}} = (\phi_w(c_1), \dots, \phi_w(c_n)) = \phi_w((c_1, \dots, c_n)).$$ \qed
\end{proof}

\subsubsection{Specialization of the multiplication matrices}

We can derive the specialization of multiplication matrices.

\begin{lemma}\label{spe mult matrices}
    Let $w \in \C^s \backslash (\mathcal{D} \cup \mathcal{W}_\mathcal{G})$. Then for any $v \in \C(W)[X]$ with no denominator vanishing at $w$ we have 
    $$\phi_w \left( M_{v}^{\mathcal{A}_W}\right) = M_v^{\mathcal{A}(w)}$$
    where $\phi_w$ is applied entry-wise. 
\end{lemma}

\begin{proof}
    For each $i \in [\![1,D]\!]$, the $i$-th column of $M_v^{\mathcal{A}_W}$ is given by the coordinates with regard to $\mathcal{B}_{\mathcal{A}_W}$ of $[vp_i]_{\mathcal{A}_W}$. We can then use Corollary \ref{spe coordinates} to get that $\phi_w (\overrightarrow{[vp_i]_{\mathcal{A}_W}}) = \overrightarrow{[vp_i]_{\mathcal{A}(w)}}$ which is the $i$-th column of $M_v^{\mathcal{A}(w)}$. \qed
\end{proof}

A direct consequence of that is the specialization of traces.

\begin{corollary}\label{spe traces}
    Let $w \in \C^s \backslash (\mathcal{D} \cup \mathcal{W}_\mathcal{G})$. Then for any $v \in \C(W)[X]$ with no denominator vanishing at $w$ we have
    $$\phi_w \left( \Trace^{\mathcal{A}_W} (v)\right) = \Trace^{\mathcal{A}(w)} (\phi_w(v)).$$
\end{corollary}

\subsubsection{Specialization of $h_0$ and $h_1$} The following lemma shows that the two first polynomials of $\mathcal{F}$ specialize well using the computations with traces. 

\begin{lemma}\label{spe carac}
    Let $w \in \C^s \backslash (\mathcal{D} \cup \mathcal{W}_\mathcal{G})$. Then we have 
    $$\phi_w \left( \chi_t^{\mathcal{A}_W}\right) = \chi_t^{\mathcal{A}(w)}$$ and in particular the Hörner polynomials of $\chi_t^{\mathcal{A}(w)}$ are obtained by applying $\phi_w$ to the ones of $\chi_t^{\mathcal{A}_W}$.

    Moreover, we have 
    $$\phi_w \left( \frac{d \chi_w^{\mathcal{A}_W}}{dT}\right) = \frac{d \chi_t^{\mathcal{A}(w)}}{dT}.$$
\end{lemma}

\begin{proof}
    The first point comes from applying both Lemma \ref{spe mult matrices} and the fact that specialization commutes with the determinant to get 
    \begin{eqnarray*}\phi_w \left(\chi_t^{\mathcal{A}_W}\right)  & = & \phi_w \left( \det (U_0I_D - M_t^{\mathcal{A}_W})\right)  \\
    & = & \det \left( TI_D - \phi_w(M_t^{\mathcal{A}_W})\right) = \det\left(U_0I_D - M_t^{\mathcal{A}(w)}\right) = \chi_t^{\mathcal{A}(w)}.
    \end{eqnarray*}

    The second part comes from the fact that specialization commutes with derivation. \qed
\end{proof}

\subsubsection{Specialization of $h_{X_1}, \dots, h_{X_n}$}

\begin{lemma}\label{spe h_Xi}
    Let $w \in \C^s \backslash (\mathcal{D} \cup \mathcal{W}_\mathcal{G})$. Then for each $i \in [\![1,n]\!]$ we have
    $$\phi_w(h_{X_i}) (U_0 ; W) = \sum_{i=0}^{D-1} \Trace^{\mathcal{A}(w)} (X_i t^i) \phi_w(H_{D-i-1}) (U_0 ; W)$$
    and $\phi_w(H_{D-i-1})$ is the $(D-i-1)$-th Hörner polynomial of $\chi_t^{\mathcal{A}(w)}$.
\end{lemma}

\begin{proof}
    This is a direct consequence of Corollary \ref{spe traces} and Lemma \ref{spe carac}.
\end{proof}

\subsubsection{Specialization of $\mathcal{F}$}

We can summarize the previous result into the following theorem.

\begin{theorem}\label{spe candidat GRUR}
    Suppose that $t \in \C[X]$ is a linear form and that the family $\mathcal{F} = \{h_0, h_1, h_{X_1}, \dots, h_{X_n}\}$ was computed using the formulas given in Lemma \ref{relations Newton} and \ref{formules h_Xi}. Then for all $w \in \C^s \backslash (\mathcal{D} \cup \mathcal{W}_\mathcal{G})$, the family $\phi_w(\mathcal{F})$ is a RUR-candidate of $V(\mathcal{I}(w))$.
    
    Moreover, if we suppose that $t$ separates the specialized variety outside of the closed Zariski subset $\mathcal{S}$ of measure 0, then for all $w \in \C^s \backslash (\mathcal{D} \cup \mathcal{W}_\mathcal{G} \cup \mathcal{S})$, $\phi_w (\mathcal{F})$ is a RUR of $V(\mathcal{I}(w))$. 
\end{theorem}

\subsection{Certifying whether $t$ is generically separating or not}

We present now a criterion to effectively check whether a given linear form $t$ is generically separating or not. The main issue  is that we cannot directly compute the polynomial $h(U_0, \dots, U_n)$. However what we can compute are specific instances $h(U_0, a_1, \dots, a_n)$ for values $a_1,\dots, a_n \in \C$. This is more than enough as shown by the following lemma.

\begin{lemma}\label{separation check}
    Let $w \in \C^s \backslash (\mathcal{D} \cup \mathcal{W}_\mathcal{G})$ such that $\Disc_{U_0}(\overline{\chi_t^{\mathcal{A}_W}})(w) \neq 0$. Then $t$ is generically separating if and only if $t$ separates the specialized variety $V(\mathcal{I}(w))$.
\end{lemma}

\begin{proof}
    If $t$ separates $V(\mathcal{I}_W)$ then we can write $\overline{\chi_t^{\mathcal{A}_W}} = \prod_{x \in V(\mathcal{I}_W)} (T - t(x))$. It is quite classical that the discriminant is then given by $\Disc_T(\overline{\chi_t^{\mathcal{A}_W}}) = \prod_{y \neq x} (t(x) - t(y))^2$, and the set of parameters for which $t$ does not separate $V(\mathcal{I}(w))$ is thus given by $V(\Disc_T (\overline{\chi_t^{\mathcal{A}_W}}))$. By hypothesis, $w$ is indeed a value outside of this variety.

    Conversely, suppose that the form $t$ separates $V(\mathcal{I}(w))$. Using the notations of §1.2, that means that the discriminant $\Disc_{U_0}(\overline{h})$ does not vanish at $(w,a)$ where $a = (a_1, \dots, a_n)$ are the coefficients of $t$. In particular, $\Disc_{U_0} (\overline{h} (U_0, a_1, \dots, a_n)) \in \C(W)$ is not identically zero. Moreover, $\overline{h}(U_0, a_1, \dots, a_n)$ is monic in $U_0$, allowing us to write 
    \begin{eqnarray*}
    \Disc_{U_0} (\overline{h}(U_0, \dots, U_n)) (a_1, \dots, a_n) & = & \Disc_{U_0} (\overline{h} (U_0, a_1, \dots, a_n)) \\
    & = & \Disc_{U_0} \left(\prod_{x \in V(\mathcal{I}_W)} (U_0 - t(x))\right) \neq 0.
    \end{eqnarray*}
    This exactly means that $t$ is injective on $V(\mathcal{I}_W)$, and thus that $t$ is generically separating. \qed
\end{proof}

In other words, we can check whether $t$ is generically separating checking whether it separates one specific specialized variety. This can done using for instance the new algorithm proposed in \cite{DRR2025}.

\subsection{About the different genericity levels}

Considering the different algebraic sets to avoid introduced in the previous sections, we find it relevant to make a summarize of their meaning. In Corollary \ref{valeurs generiques GRUR} we showed that the GRUR can be specialized outside $\mathcal{D} \cup \mathcal{L} \cup \mathcal{W}_\infty \cup \mathcal{S}$. 

\begin{itemize}
    \item Outside $\mathcal{D} \cup \mathcal{W}_\infty \cup \mathcal{L} $, the specialized systems have a constant number of roots counted with multiplicities. Their number is $D$ which is the degree of $h_0$ or equivalently the dimension of the quotient $\C(W)[X]/\mathcal{I}_W$ as a $\C(W)$-vector space.
    \item Outside $\mathcal{D} \cup \mathcal{W}_\infty \cup \mathcal{S} \cup \mathcal{L} $, the specialized systems have a constant number of distinct roots. Their number is $d$ which is the degree of $\overline{h_0}$ when $t$ is generically separating. 
\end{itemize}

In Theorem \ref{spe candidat GRUR} we showed that if we compute $\mathcal{F}$ using linear algebra in the quotient $\C(W)[X]/\mathcal{I}_W$, specialization holds for parameters outside $\mathcal{W}_\mathcal{G}$. This is the condition for the computations in the quotient $\mathcal{A}_W$ to specialize but they are just raised by the algorithmic way of computing $\mathcal{F}$.

\subsection{Algorithm}

We propose here a first Las-Vegas algorithm to compute $\mathcal{F}$. The sub-algorithms can be found in the Appendix. We do all the computations with a fixed linear form chosen in the set $\mathcal{T} := \{X_1 + kX_2 + \dots + k^{n-1}X_n, ~ k \in [\![1, D(D-1)/2]\!]\}$ which contains at least one linear form that separates $V(\mathcal{I}_W)$ (see \cite[Lemma 2.1]{Rouillier1998}). The so-called traces vector $\VTr 1$ refers to the vector $(\Trace^{\mathcal{A}_W} (p_1), \dots, \Trace^{\mathcal{A}_W} (p_D))^T$ which turns out to simplify some computations.

\begin{algorithm}[H]
\caption{GRUR-LA}\label{GRUR-LA}
\begin{algorithmic}
\Require $\mathcal{G}$ a reduced Gröbner basis of $\mathcal{I}_W$ w.r.t $<_X$ ; a basis $\mathcal{B}_{\mathcal{A}_W}$ of $\mathcal{A}_W$ ; the multiplication matrices $M_{X_1}, \dots, M_{X_n}$ by $X_1, \dots, X_n$ in $\mathcal{A}_W$ ; $\MT(\mathcal{A}_W)$ the multiplicative tensor of $\mathcal{A}_W$.
\Ensure $\mathcal{F} = \{h_0, h_1, h_{X_1}, \dots, h_{X_n}\}$ a generic RUR of $V(\mathcal{I}_W)$.
\State \underline{Step 0 :} Compute $\VTr 1$ the traces vector
\State \underline{Step 1 :} Choose a not already chosen linear form $t$ in $\mathcal{T}$
\While{$t$ is not generically separating}
    \State $h_0 \gets \text{CharacteristicPolynomial}(\MT(\mathcal{A}_W), M_t, \VTr 1)$
    \State Test whether $t$ is separating or not using DRR
    \If{$t$ is separating}
        \State \Return{$\mathcal{F} = \{\chi_t, \chi_t', \text{GRURPolynomials}(h_0, M_{X_1}, \dots, M_{X_n}, \VTr 1)\}$}
     \Else
     	\State Go back to Step 1.
    \EndIf

\EndWhile
\end{algorithmic}
\end{algorithm}

\section{Second algorithm : evaluation/interpolation scheme}

\subsection{Presentation of the algorithm}

We propose here a second algorithm for computing $\mathcal{F}$. The reason is that Algorithm \ref{GRUR-LA} requires the computation of a Gröbner basis of the ideal generated by $f_1, \dots, f_n$. This Gröbner basis depends on parameters in a rather uncontrollable way and there is a high chance that the degree in the parameters explodes. In addition to that, computing the characteristic polynomial of $t$ requires the computation of $\Trace^{\mathcal{A}_W} (t^i)$ for $i \in [\![1, D]\!]$ whose memory size might explode even though we know that the final result has a controlled size (see the bounds of Corollary \ref{bounds RUR}). For these reasons, we propose a second approach using an evaluation/interpolation scheme.

The proposed algorithm is presented below. Each step will be precisely studied through this section.

\begin{algorithm}[H]\label{GRUR-EI}
    \caption{\textit{GRUR-EI}}
    \begin{algorithmic}
    \Require Polynomials $f_1, \dots, f_n \in \Q[W_1, \dots, W_s,X_1, \dots, X_n] = \Q[W,X]$ that form a generically zero-dimensional system with degree in $X$ bounded by $d_X$ and degree in $W$ bounded by $d_W$. 
    \Ensure $\mathcal{F} = \{h_0, h_1, h_{X_1}, \dots, h_{X_n}\} \subseteq \Q(W)[T]$ a generic RUR of $V(\mathcal{I}_W)$.
    \State \underline{Step 0 :} Compute a grid $\mathcal{A} \subseteq \Q^s$ of at least $((n+1)d_X^n d_W + 1)^s$ of points $w \notin (\mathcal{D} \cup \mathcal{L} \cup \mathcal{S})$.
    \State \underline{Step 1 :} Choose a not already chosen linear form $t \in \mathcal{T}$.
    \State \underline{Step 2 :} For each $w \in \mathcal{A}$ use the algorithm \textit{BlackBoxRUR} to compute specialized RURs of the system and eventually check if $t$ is generically separating.
    \If{$t$ is not generically separating} 
        \State Go back to Step 1.
    \EndIf
    \State \underline{Step 3 :} Using the computed RURs interpolate $h_0, h_1, h_{X_1}, \dots, h_{X_n}$.

    \noindent
    \Return{$\mathcal{F} = \{h_0, h_1, h_{X_1}, \dots, h_{X_n}\}$}
    \end{algorithmic}
\end{algorithm}

\begin{remark}
    In practice, it is even better to first try with $t$ being one of the variables. If none is working, then we start taking elements from $\mathcal{T}$. Moreover, experiments tend to show that we only need a few tries to find a generically separating form. 
\end{remark}

\subsection{Complexity analysis of Algorithm \textit{GRUR-EI}}

\subsubsection{Complexity of Step 0}

Thanks to Corollary \ref{bounds RUR}, we know that the polynomials of the GRUR have degree in $W$ bounded by $\kappa := (n+1)d_X^n d_W$. Therefore, we need at least $\binom{\kappa+s}{s} \leqslant (\kappa + 1)^s$ interpolation points. The idea is to find an estimate on the degrees of the polynomials defining $\mathcal{D}$, $\mathcal{L}$ and $\mathcal{S}$ and then construct a grid in $\Q^s$ that contains enough good interpolation points.

\begin{proposition}
     Let $q := \kappa' + nd_X^{n-1} d_X + n d_W$. The grid $\mathcal{A} = A_{k_0}^s$ where $k_0 = \lceil q+(\kappa+1)^s/q^{s-1} \rceil$ is formed of $O(n^s d_X^{2ns} d_W^s)$ points in $\Q$ and contains at least $(\kappa + 1)^s$ values $w \in \Q^s \backslash (\mathcal{D} \cup \mathcal{L} \cup \mathcal{S})$.
\end{proposition}

\begin{proof}
For $\mathcal{D}$, we recall that $\mathcal{D} \subseteq \mathcal{W}_\infty = V(C)$ where $C = \Res(\tilde{F_1}, \dots, \tilde{F}_n)$ and $\tilde{F}_i = F_i(0, X_1, \dots, X_n)$ (see Remark \ref{generic sets}). So that whenever $C(w) \neq 0$, the specialized system at $w$ has a finite number of solutions. We also recall that $\deg_W(C) \leqslant nd_X^{n-1}d_W$.

For $\mathcal{L}$, denote by $L(W) := \prod_{i=1}^n \LT_{<_X} (f_i) (W)$. Then $\mathcal{L} = V(L)$ and $L$ has degree in $W$ bounded by $n d_W$.

For $\mathcal{S}$, we write $h_0 (U_0 ; W) = \Delta(U_0 ; W)/p(W)$ with $\Delta \in \Z[U_0, W]$ and $p \in \Z[W]$. Recall that whenever $t$ is generically separating, $\mathcal{S} = V(\Disc_{U_0} (\overline{\Delta}))$. Since $\deg_W(\Delta) \leqslant \kappa = (n+1)d_X^n d_W$, \cite[Proposition 8.45]{BPR2006} explains that this discriminant has degree in $W$ bounded by $\kappa' := 2(n+1)d_X^{2n}d_W$.

Now to find enough good interpolation points, we build a grid $A_k \times \dots \times A_k = A_k^s$ where $A_k = [\![1,k]\!]$ and we look for a value $k_0$ such that $A_{k_0}^s$ contains enough parameters values for which neither $C(W)$, $L(W)$ nor $\Disc_{U_0}(\overline{\Delta}) (W)$ vanish. The main ingredient is Schwartz-Zippel theorem (\cite[Proposition 97, §12]{Zippel2012}). It states that a non-zero polynomial $F \in \Q[X_1, \dots, X_n]$ of total degree $d$ vanishes at most $d |S|^{n-1}$ times on the grid $S^n$. 

Here, we take $F$ to be $\Disc_T(\overline{\Delta}) \cdot C \cdot L \subseteq \Z[W]$ which has degree in $W$ bounded by $\kappa' + nd_X^{n-1} d_W + n d_W$. Denote by $q := \kappa' + nd_X^{n-1} d_X + n d_W$.We want to find $k_0 \geqslant 0$ such that $|A_{k_0}|^s - q |A_{k_0}|^{s-1} \geqslant (\kappa+1)^s$. Consider $g : x \mapsto x^s - q x^{s-1} - (\kappa+1)^s = x^{s-1}(x - q) - (\kappa+1)^s$. Clearly, $g(x) \leqslant 0$ for $x \in [0,q]$ so $k_0 \geqslant q$. In addition to that, $g$ is increasing on $[q, +\infty[$. Denote by $x_0$ the solution to the equation $g(x) = 0$. Then for all $x \geqslant x_0$ we have $g(x) \geqslant 0$. Moreover, $x_0 > q$ so that $x_0^{s-1} > q^{s-1}$ and then $x_0^{s-1} (x_0 - q) = (\kappa+1)^s > q^{s-1} (x_0 - q)$.
In other words, we have shown that $x_0 < q+(\kappa+1)^s/q^{s-1}$ and we thus choose $k_0 := \lceil q+(\kappa+1)^s/q^{s-1} \rceil$ where $\lceil \cdot \rceil$ is the ceil function. For complexity purposes, observe that $k_0 = O(q) = O(nd_X^{2n} d_W)$. \qed
\end{proof}

\subsubsection{Complexity of Step 2}

First, we describe the algorithm \textit{BlackBoxRUR} and give an estimate of its complexity. 

\begin{algorithm}[H]\label{BlackBoxRUR}
    \caption{\textit{BlackBoxRUR}}
    \begin{algorithmic}
    \Require Polynomials $g_1, \dots, g_n \in \Q[X_1, \dots, X_n]$ that form a zero-dimensional system, and a linear form $t \in \Q[X_1, \dots, X_n]$.
    \Ensure $\mathcal{F} = \{h_0, h_1, h_{X_1}, \dots, h_{X_n}\} \subseteq \Q[T]$ a RUR of $V(g_1, \dots, g_n)$.
    \State \underline{Step 1 :} Compute a Gröbner basis $\mathcal{G}$ of $\langle g_1, \dots, g_n \rangle$ for the ordering <$X$.
    \State \underline{Step 2 :} Compute the multiplication matrices $M_{X_1}, \dots, M_{X_n}$ of the multiplication by each variables in the quotient $\Q[X]/\langle g_1, \dots, g_n \rangle$.
    \State \underline{Step 3 :} If $t$ separates $V(g_1, \dots, g_n)$, compute $\mathcal{F}$ a RUR of $V(g_1, \dots, g_n)$ with regard to $t$.
    
    \noindent
    \Return{$\mathcal{F}$ or "False"}
    \end{algorithmic}
\end{algorithm}

\begin{proposition}
    A single call of \textit{BlackBoxRUR} for polynomials of degrees bounded by $d$ requires at most $\tilde{O}(nd^{3n})$ operations in $\Q$.
\end{proposition}

\begin{proof}
    We study separately every step of the algorithm.

    $\cdot$ \underline{Step 1 :} The key ingredient here is the following lemma that we directly prove. We recall that a family $g_1, \dots, g_n$ is a \textit{regular sequence} if for each $i \in [\![2, n]\!]$, $g_i$ is not a zero-divisor in the quotient $\Q[X]/\langle g_1, \dots, g_{i-1} \rangle$. 

    \begin{lemma}\label{regular}
       Let $g_1, \dots, g_n$ be a sequence of elements of $\Q[X_1, \dots, X_n]$ such that $V(\langle g_1, \dots, g_n\rangle)$ has dimension zero. Then the sequence $g_1, \dots, g_n$ is regular.
    \end{lemma}
    \begin{proof}
    To prove this, we use height theory for ideals (see for instance \cite{Matsumura1980}). Since $\Q[X]$ is a Cohen-Macaulay ring, proving that the sequence $g_1, \dots, g_n$ is regular is equivalent to showing that for each $i \in [\![1,n]\!]$ we have $\text{ht}(\langle f_1, \dots, f_i \rangle) = i$ where $\text{ht}$ denotes the height. In what follows, we denote by $\mathcal{J}_i := \langle g_1, \dots, g_i \rangle$, $1 \leqslant i \leqslant n$.

    Let $\mathfrak{p} \supseteq \mathcal{J}_n$ be a prime ideal. Then we have 
    \begin{equation}\label{eq dim}
    \text{ht}(\mathfrak{p}) + \dim_{Krull} \left( \frac{\Q[X]}{\mathfrak{p}} \right) = \dim_{Krull} (\Q[X]) = n
    \end{equation}
    where $\dim_{Krull}$ denotes the Krull dimension (see \cite[Theorem 1.8A]{Hartshorne1977}). Since $\mathfrak{p} \supseteq \mathcal{J}_n$ we have that $V(\mathfrak{p}) \subseteq V(\mathcal{J}_n)$. Then $V(\mathfrak{p})$ is also finite, and thus $\Q[X]/\mathfrak{p}$ has Krull dimension zero which means that Equation \eqref{eq dim} yields $\text{ht}(\mathfrak{p}) = n$ and thus that $\text{ht}(\mathcal{J}_n) = n$.

    In addition to that, Krull Principal Ideal Theorem (see \cite[Theorem 18 §12]{Matsumura1980}) states that for each $i \in [\![1,n]\!]$, $\text{ht}(\mathcal{J}_i) \leqslant i$. Adding a generator to an ideal can increase its height by at most one, so having $\text{ht}(\mathcal{J}_n) = n$ forces each of theses inequalities to be an equality. \qed
    \end{proof}

    During a call of \textit{BlackBoxRUR}, the sequence $g_1, \dots, g_n$ satisfies the assumptions of lemma \ref{regular} and is thus a regular sequence. The authors in \cite[Proposition 6]{BFS2005} show that the complexity of computing a Gröbner basis of $\langle g_1, \dots, g_n \rangle$ lies in $O(\binom{n + d_{reg}}{n}^\omega)$, with $\omega$ the exponent of linear algebra. A usual bound on $d_{reg}$ is the Macaulay bound $d_{reg} \leqslant nd - n + 1$. In other words, the Gröbner basis of $\langle g_1, \dots, g_n \rangle$ for the ordering $<_X$ can be computed within $O(\binom{n + d_{reg}}{n}^\omega) = O(nd^{wn})$ operations in $\Q$.

    $\cdot$ \underline{Step 2 :} With the knowledge of a Gröbner basis of $\langle g_1, \dots, g_n \rangle$, it is possible to compute the multiplication matrices $M_{X_1}, \dots, M_{X_n}$ with $O(nd^{3n})$ operations in $\Q$ (see \cite[Proposition 3.1]{FGLM1993}).

    $\cdot$ \underline{Step 3 :} For this step we refer to the algorithm proposed in \cite[Algorithm 6]{DRR2025}. Given the multiplication matrices $M_{X_1}, \dots, M_{X_n}$ of the multiplication by the variables in $\Q[X]/\langle g_1, \dots, g_n \rangle$ and the linear form $t$, this algorithm will either return "False" if $t$ does not separate $V(\langle g_1, \dots, g_n \rangle)$ or will compute a (reduced) RUR of the radical with regard to $t$ otherwise with $\tilde{O}(nd^{3n})$ operations in $\Q$. The passage from this reduced RUR to a RUR of $\langle g_1, \dots, g_n \rangle$ can be done by replacing the first polynomial of the reduced RUR of $V(\sqrt{\langle g_1, \dots, g_n \rangle})$ by the characteristic polynomial of $t$ (which can be computed with $\tilde{O}(d^{3n})$ operations in $\Q$). The overall complexity of this step then lies in $\tilde{O}(nd^{3n})$.

    Thus, the global complexity of algorithm \textit{BlackBoxRUR} lies in $\tilde{O}(nd^{3n})$. \qed
\end{proof}

\begin{corollary}
   Step 2 in Algorithm \textit{GRUR-EI} requires $\tilde{O}(n^{s+2} d_X^{2ns + 3n} d_W^s (d_X^{2n} + d_W^s))$ operations in $\Q$.
\end{corollary}

\begin{proof}
For each value of parameters $w \in \mathcal{A}$, we want to specialize the $f_i'$s to $w$. One specialization can be done using $O(nd_W^s d_X^n)$. Indeed, each $f_i$ can be written as $f_i = \sum_{|\alpha| \leqslant d_X} c_{i, \alpha} X^{\alpha}$ with $c_{i, \alpha} \in \Q[W]$ and $\deg_W (c_{i, \alpha}) \leqslant d_W$. We can evaluate all monomials in $W$ of degree less or equal than $d_W$ using $\binom{d_W + s}{s}$ operations in $\Q$ by successively computing them in degree increasing order. As each $f_i$ is made of at most $\binom{d_X + n}{n}$ monomials in $X$, the total number of operations to compute $f_1(w, X), \dots, f_n (w, X)$ lies in $O(n \binom{d_W + s}{s} \binom{d_X + n}{n}) = O(nd_W^s d_X^n)$.

We need to perform at most $\binom{\kappa + s}{s} + (\kappa' + nd_X^{n-1}d_W + n d_W)k_0^{s-1} = O(n^s d_X^{2ns} d_W^s)$ evaluations. The total cost of these evaluations thus lies in $O(n^{s+1} d_X^{2ns + n} d_W^{2s} )$. Then, for each of these specialized systems we use the algorithm \textit{BlackBoxRUR} for $\tilde{O}(nd_X^{3n})$ additionnal operations in $\Q$. Computing all specialized RURs requires $\tilde{O} (n^{s+1} d_X^{2ns + 3n} d_W^s)$ operations in $\Q$. The cost of evaluating all the systems and using the black-box for each value $w \in \mathcal{A}$ therefore requires at most $\tilde{O}(n^{s+1}d_X^{2ns + n} d_W^s (d_X^{2n} + d_W^s))$ operations in $\Q$. 

Now, this part of Step 2 runs until a generically separating form is found. The set $\mathcal{T}$ from which the candidates are chosen has size bounded by $nd_X^n(d_X^n-1) = O(nd_X^{2n})$ and is guaranteed to contain at least one generically separating form so the preceeding computations are done at most $O(nd_X^{2n})$ times. The final complexity then lies in $\tilde{O}(n^{s+2} d_X^{2ns + 3n} d_W^s (d_X^{2n} + d_W^s))$. \qed
\end{proof}

\begin{remark}
    The result from Schwartz and Zippel guarantees that on the grid $\mathcal{A}$ there are at most $\kappa' k_0^{s-1}$ points on which $\Disc{U_0}(\overline{h}_0)$ can vanish. If during the computations of Step 2, \textit{BlackBoxRUR} returns more than $\kappa' k_0^{s-1}$ times "False" then we know that $\Disc_{U_0} \overline{h}_0$ is identically zero. In that case, $t$ is not generically separating and the algorithm \textit{GRUR-EI} should go back to step 1.  
\end{remark}

\subsubsection{Complexity of Step 3}

We denote by $I(n,d)$ a bound on the number of operations in $\Q$ to interpolate a polynomial in $n$ variables of total degree $d$ from its values on $\binom{d+n}{n}$ points.

\begin{proposition}\label{cost interpolation}
    Step 3 of algorithm \textit{GRUR-EI} requires $O(n d_X^n I(\kappa,s))$ operations in $\Q$.
\end{proposition}

\begin{proof}
We need to interpolate the $2(D+1)$ numerators and denominators of $h_0$, and the $2(n+1)D$ numerators and denominators of $h_1, h_{X_1}, \dots, h_{X_n}$. All have degree in $W$ bounded by $\kappa$. That means the interpolation part requires $O(n d_X^n I(\kappa, s))$ operation in $\Q$. \qed
\end{proof}

\subsubsection{Global complexity}

We summarize the complexity result into the following theorem.

\begin{theorem}\label{cost GRUR-EI}
    In Algorithm \textit{GRUR-EI},
    \begin{itemize}
    	\item the evaluation part requires at most $\tilde{O}(n^{s+2} d_X^{2ns + 3n} d_W^s (d_X^{2n} + d_W^s))$ operations in $\Q$ ;
	\item the interpolation part requires at most $O(nd_X^n I(\kappa, s))$ operations in $\Q$. 
    \end{itemize}
    Therefore, Algorithm \textit{GRUR-EI} can be performed using $\tilde{O}(n^{s+2} d_X^{2ns + 3n} d_W^s (d_X^{2n} + d_W^s) + nd_X^n I(\kappa, s))$ operations in $\Q$.
\end{theorem}

We can now give a more precise expression using known complexities for interpolation. 

\begin{corollary}
	Algorithm \textit{GRUR-EI} can be performed using at most 
	$$\tilde{O}(sn^{2s+2} d_X^{2ns + 2n} d_W^{2s+1} + n^{s+1}d_X^{2ns + 5n} d_W^s)$$
	operations in $\Q$.
\end{corollary}

\begin{proof}
The algorithm proposed in \cite{KY1989} allows to interpolate a polynomial in $n$ variables, with total degree $d$ and at most $\theta$ terms with $O(dn\log_2(n) \theta M(\theta)\log_2(\theta))$ arithmetic operations. Here $M(\theta)$ is the complexity bound for multiplying two univariate polynomials of degree less than $\theta$. We can take the classical Schönhage bound with $M(\theta) = O(\theta \log_2(\theta) \log_2 \log_2(\theta))$. In our case, $d \leqslant \kappa$ and $\theta \leqslant (\kappa + 1)^s$, and with our notations, the cost of a single interpolation is $I(\kappa, s) = \tilde{O} (sn^{2s+1} d_X^{2ns + n} d_W^{2s+1})$ arithmetic operations. 
By substituting that in the complexity found in Theorem \ref{cost GRUR-EI} we get that the GRUR polynomials can be computed with $\tilde{O} \big(sn^{2s+2} d_X^{2ns + 2n} d_W^{2s+1}\big)$ operations in $\Q$. \qed
\end{proof}

\subsection{Real roots classification} 

This problem consists in finding polynomials formulas describing semi-algebraic sets $\mathcal{S}_1, \dots, \mathcal{S}_l \subseteq \R^s$ such that their union is dense in $\R^s$, and above each of them, the specialized systems have a constant number of real roots. In addition to that, we want to be able to compute one point in each of these connected component.

In \cite{LSeD2022}, the authors tackle this problem by computing the parametric Hermite matrix, that is, the matrix $\mathcal{H} := \left(\Trace^{\mathcal{A}_W}(p_i p_j)\right)_{1 \leqslant i, j \leqslant D}$ where we recall that $\mathcal{B} = \{[p_1]_{\mathcal{A}_W}, \dots, [p_D]_{\mathcal{A}_W}\}$ is a basis of $\mathcal{A}_W$. The authors show that the looked for semi-algebraic sets can be described by the principal minors of the Hermite matrix as they encode the parameters values for which the signature of $\mathcal{H}$ changes (and thus the number of distinct real roots). Their algorithm is general and they show that it runs with at most $\tilde{O} \left( \binom{s+\mathfrak{D}}{s} 2^{3s} n^{2s + 1} d^{3ns + 2(n+s) + 1}\right)$ where $d$ is the total degree of the equations and $\mathfrak{D} = n(d-1)d^n$. To reach this complexity, the author make some assumptions on the Gröbner basis computed for the block ordering $<_{X, W}$ that give a great control on the potential explosion of the degrees in the parameters during the computation of $\mathcal{H}$ (see \cite[Assumptions B and C]{LSeD2022}). We propose here a general approach with an estimated higher complexity that do not depend on these hypothesis.

The formulas describing the semi-algebraic sets $\mathcal{S}_1, \dots, \mathcal{S}_l$ can be derived from the computation of a Sturm-Habicht (or Sylvester-Habicht) sequence, or rather the leading coefficients of this sequence. We stick to the definition given in \cite[Notation 4.21]{BPR2006} and our main tool will be \cite[Theorem 3]{GLRR1989} to count to number of real roots of the system. Given $\mathbb{A}$ a ring and $P, Q \in \mathbb{A}[U_0]$, denote by $\textbf{sres} (P,Q)$ the \textit{signed subresultant sequence} of $P$ and $Q$. The \textit{Sturm-Habicht sequence} of $P$ is given by $\textbf{sres}(P, P')$.

The fact that $\mathcal{F} = \{h_0, h_1, h_{X_1}, \dots, X_n\}$ forms a generic RUR of $V(\mathcal{I}_W)$ can be rephrased as the fact that for almost all $w \in \C^s$, the morphism of algebraic sets 
$$\fonc{\Psi_w}{V(\phi_w(h_0))}{V(\mathcal{I}(w))}{y}{\left(\frac{\phi_w(h_{X_1})(y)}{\phi_w(h_1) (y)}, \dots, \frac{\phi_w(h_{X_n}) (y)}{\phi_w(h_1)(y)}\right)}$$
preserves multiplicities, and also the number of real roots. In other words, the number of real roots of $V(\mathcal{I}(w))$ is exactly the number of real roots of $\phi_w(h_0) = \chi^{\mathcal{I}(w)}$. We recall that $h_0$ can be computed under the form $h_0 = \Delta/P$ with $\Delta \in \Z[W][U_0]$ and $P \in \Z[W]$. The number of real roots of $h_0$ is the same as the number of real roots of $\Delta$.

Since $\phi_w(\Delta)$ is a univariate polynomial, \cite[Theorem 3]{GLRR1989} guarantees that the number of its real roots is given by the number of sign changes (see \cite[§1]{GLRR1989}) in the sequence $\textbf{sres}(\Delta), \phi_w(\Delta)')$ computed with regard to the variable $U_0$. The great advantage of using subresultants is that they have really great specialization properties. Indeed, for each $w \in \C^s$ such that the leading coefficient of $\Delta$ in $U_0$ does not vanish at $w$, the Sturm-Habicht sequence specializes in the sense that $\phi_w \big(\textbf{sres}(\Delta, \Delta') \big) = \textbf{sres}\big( \phi_w(\Delta), \phi_w(\Delta)' \big)$.

The Sturm-Habicht sequence $\textbf{sres}(\Delta, \Delta')$ is made of polynomials of $\Z[W]$ that describe the semi-algebraic sets $\mathcal{S}_1, \dots, \mathcal{S}_l \subseteq \R^s$ on which the number of real roots of $V(\mathcal{I}(w))$ is constant. The excluded parameters values are those for which the GRUR does not specialize, that is those contained in $\mathcal{D} \cup \mathcal{L} \cup \mathcal{W}_\infty \cup \mathcal{S}$ (see Corollary \ref{valeurs generiques GRUR}) and those for which the leading coefficient of $\Delta$ vanishes. The union of these algebraic sets is of measure zero so the union of the $S_i$'s is indeed dense in $\R^s$. This description leads to the following algorithm whose complexity is studied in Proposition \ref{complexite rrc}.

\begin{algorithm}[H]\label{RealRootClassification}
    \caption{RealRootClassification}
    \begin{algorithmic}
    \Require Polynomials $f_1, \dots, f_n \in \Z[W,X]$ generating a generically zero-dimensional ideal $\mathcal{I}$.
    \Ensure Polynomials $P_1, \dots, P_D \in \Z[W]$ that are the formulas of the semi-algebraic sets solving the real roots classification problem along at least one point in each of them. 

    \State \underline{Step 1 :} Compute $\mathcal{F} = \{h_0, h_1, h_{X_1}, \dots, h_{X_n}\}$ a GRUR of $V(\mathcal{I}_W)$ with $h_0 = \Delta/P$.
    \State \underline{Step 2 :} The subresultants sequence $\textbf{sres}\big(\Delta, \Delta' \big) = \{P_1, \dots, P_D\}$.
    \State \underline{Step 3 :} Compute a set $\mathcal{P}$ of at least one point per connected component of the semi-algebraic set $\Q = \{P_1 \neq 0, \dots, P_D \neq 0\}$.

    \Return{$P_1, \dots, P_D$ and $\mathcal{P}$}
    \end{algorithmic}
\end{algorithm}

\begin{proposition}\label{complexite rrc}
    Given polynomials $f_1, \dots, f_n \in \Z[W, X]$ generating a generically zero-dimensional ideal $\mathcal{I}$, Algorithm \ref{RealRootClassification} solves the real roots classification problem and can be performed using at most
    $$\tilde{O} \big( 2^{6s+1} n^{3s+1} d_X^{7ns+3n} d_W^{3s+1} \big)$$
    operations in $\Q$.
\end{proposition}

\begin{proof}
    The discussion held in the previous subsection showed that we can compute the elements of a GRUR of $V(\mathcal{I}_W)$ with at most $\tilde{O} \big(sn^{2s+2} d_X^{2ns + 2n} d_W^{2s+1}\big)$ operations in $\Q$. That represents the cost of Step 1. We write $h_0 = \Delta/P$ where $\Delta \in \Q[W,T]$ and $P \in \Q[W]$. 
    
    Next, we need to compute the Sturm-Habicht sequence $\textbf{sres} (\Delta, \Delta')$. Thanks to the bounds found in corollary \ref{bounds RUR}, we know that $\Delta$ has total degree bounded by $\kappa + d_X^n$ and height bounded by $\tilde{O} \big(nd_X^n \tau + nd_X^n d_W \log_2(s+1)\big)$. The idea is to use binary segmentation to transform $\Delta$ and $\Delta'$ into polynomials in $\Z[Y][T]$, compute the subresultants sequence of the transformed polynomials and then recover the result in $\Z[W]$ using inverse binary segmentation. To do that, we need bounds on both the height and degrees of the expected result. From \cite[proposition 8.45]{BPR2006} we know that each polynomial of $\textbf{sres}(\Delta, \Delta')$ has degree bounded by $\kappa' := 2(n+1) d_X^{2n} d_W$.

    The ring homomorphism $\Theta : \Z[W][T] \rightarrow \Z[Y][T]$ such that $W_1 \mapsto Y$, $W_2 \mapsto Y^{\kappa' + 1}$, $\dots$, $W_s \mapsto Y^{\kappa'^{s-1} + 1}$ where $T$ is a new variable allows us to obtain $\textbf{sres}(\Delta, \Delta')$ through the computation of $\textbf{sres}(\Theta(\Delta), \Theta(\Delta')) = \Theta(\textbf{sres}(\Delta, \Delta'))$ (see for example \cite[§3]{DET2009}, \cite[Theorem 3.3]{Klose1995}).
    Since $\Delta$ has degree in $W$ bounded by $\kappa$, we know that $\Theta(\Delta)$ has degree in $Y$ bounded by $\kappa(\kappa'^{s-1} + 1)$, that is, $\deg_Y(\Delta) = O(2^{s-1} n^s d_X^{2ns - n} d_W^{s})$. According to \cite[Corollary 11.18]{GG2003}, we can compute $\textbf{sres}(\Theta(\Delta), \Theta(\Delta'))$ using $\tilde{O} (2^{s-1} n^s d_X^{2ns} d_W^s)$ operations in $\Q$. From that, we can decode the leading coefficients of every polynomial in $\textbf{sres}(\Theta(\Delta), \Theta(\Delta'))$ and obtain polynomials $P_1, \dots, P_D \in \Z[W]$. The overall cost of Step 2 is thus $\tilde{O} (2^{s-1} n^s d_X^{2ns} d_W^s)$ operations in $\Q$.
    
    Now, \cite[Theorem 3]{GLRR1989} states that for each parameters value $w \in \R^s$, the number of real roots of $\Delta$ is given by the modified number of sign variations. This number changes if and only if some $P_i$ vanish when evaluated at $w$. The semi-algebraic $\mathcal{Q}$ defined by $P_1 \neq 0, \dots, P_D \neq 0$ is such that the number of real roots of $\Delta$ is constant on each of its connected components.

    The final step is to actually compute one point per connected component of $\mathcal{Q}$. This can be done using for instance the probabilistic algorithm proposed in \cite[Corollary 3]{LSeD2022}. The latter requires $\tilde{O} \big( 2^{6s+1} n^{3s+1} d_X^{7ns+2n + s} d_W^{3s^2+1} \big)$ operations in $\Q$ to compute a set of at most $O(4^s n^s d_X^{3ns} d_W^s)$ points in $\Q^s$ that meets at least one time each connected component of $\mathcal{Q}$.

    The complexity of the overall approach is dominantly held by the very last step. \qed
\end{proof}

\begin{remark}
	We only need $h_0$ to perform Algorithm \ref{RealRootClassification}. However, since the overall complexity is held by the last step, that would not change the complexity bound.
\end{remark}

\begin{remark}
	One could use the half-gcd algorithm as presented in \cite[Algorithm 4]{Lecerf2019} and reach a complexity lying in $\tilde{O}(n^s d_X^{2ns + n} d_W^s)$ for Step 2 by using the known bound on the degrees in $W$ appearing in the subresultants sequence. However, since the complexity is dominantly held by the last step, that would not change the overall complexity.
\end{remark}

\section{Appendix}

\subsection{Stickelberger's theorem}

For simplicity, we refer to this result as Stickelberger's theorem. It can be found in \cite[§2.3.1 Corollary 3.6]{CCS1999}.

\begin{theorem}[Stickelberger]\label{Stickelberger}
	Let $K$ be a field and $\mathcal{I} \subseteq K[X_1, \dots, X_n]$ a zero-dimensional ideal. Let $m_p$ be the endomorphism of multiplication by $p$ in the quotient $K[X_1, \dots, X_n]/\mathcal{I}$. Then :
	\begin{itemize}
		\item the characteristic polynomial of $m_p$ is $\prod_{x \in V(\mathcal{I})} (T - p(x))^{\mu(x)}$ ;
		\item the determinant of $m_p$ is $\prod_{x \in V(\mathcal{I})} p(x)^{\mu(x)}$ ;
		\item the trace of $m_p$ is $\sum_{x \in V(\mathcal{I})} \mu(x) p(x)$ ;
	\end{itemize} 
	where $\mu(x)$ is the multiplicity of $x$.
\end{theorem}

\subsection{Sub-algorithms for \textit{GRUR-LA}}

\begin{algorithm}[H]\label{base quotient}
    \caption{QuotientBasis}
    \begin{algorithmic}
    \Require $\mathcal{G}$ a Gröbner basis of $\mathcal{I}_W$ w.r.t. $\grevlex(X)$.
    \Ensure $\mathcal{B}$ a basis of $\mathcal{A}_W$.
    \If{$1$ is reducible by $\mathcal{G}$}
    \Return{$\emptyset$} \EndIf

    \State $\mathcal{B} \gets \{[1]\}$
    \State $L \gets [X_1, \dots, X_n]$ ; $L' \gets [1]$
    \While{$L' \neq \emptyset$}
        \State $L' \gets \emptyset$
        \For{$m \in L$}
            \If{$m$ is not reducible by $\mathcal{G}$}
                \State $\mathcal{B} \gets \mathcal{B} \cup \{[m]\}$
                \State $L' \gets L' \cup \{m\cdot X_1, \dots, m \cdot X_n\}$
                \State $L \gets L'$
            \EndIf
        \EndFor
    \EndWhile

    \noindent
    \Return{$\mathcal{B}$}
    \end{algorithmic}
\end{algorithm}

\begin{algorithm}[H]\label{MultiplicationMatrices}
    \caption{MultiplicationMatrices}
    \begin{algorithmic}
    \Require $\mathcal{G}$ a Gröbner basis of $\mathcal{I}_W$ w.r.t. $<_X$ ; $B(\mathcal{G})$ the set of monomials from a basis $\mathcal{B}_{\mathcal{A}_W}$ of $\mathcal{A}_W$. 
    \Ensure The multiplication matrices $M_{X_1}, \dots, M_{X_n}$ by $X_1, \dots, X_n$ in $\mathcal{A}_W$.
    \State Compute $M(\mathcal{G})$ and order $B(\mathcal{G}) \cup M(\mathcal{G})$ using $<_X$. 
    \For{$m \in B(\mathcal{G}) \cup M(\mathcal{G})$}
        \If{$m \in B(\mathcal{G})$}
            \State Write $m = p_i$ for some $i \in [\![1,D]\!]$. Then $\overrightarrow{[m]} \gets [\delta_{l,i},~ 1 \leqslant l \leqslant D]^T$.
        \EndIf
        \If{$m = \LT_{<X} (g)$ for some $g \in \mathcal{G}$ written $g = m + \sum_{l=1}^D a_l p_l$}
            \State $\overrightarrow{[m]} \gets [-a_1, \dots, -a_D]^T$. 
        \Else
            \State Write $m = X_j m'$ with $m' = X_s b$. 
            \For{$u = 1 \dots D$}
                \State $\overrightarrow{[m]_u} \gets \sum_{v=1}^D \overrightarrow{[X_s b]_v} \overrightarrow{[X_j p_v]_u}.$
            \EndFor
        \EndIf
    \EndFor
    \For{i= 1 \dots n}
        \State Build $M = \begin{pmatrix} \overrightarrow{[X_i p_1]} \mid \cdots \mid\overrightarrow{[X_i p_D]}\end{pmatrix}$
    \EndFor

    \noindent
    \Return{$M_{X_1}, \dots, M_{X_n}$}
    \end{algorithmic}
\end{algorithm}

\begin{algorithm}[H]\label{MultiplicativeTensor}
    \caption{MultiplicativeTensor}
    \begin{algorithmic}
    \Require $\mathcal{B}$ a basis of $\mathcal{A}_W$ and $M_{X_1}, \dots, M_{X_n}$.
    \Ensure The multiplicative tensor $\MT(\mathcal{A}_W)$.
    \State $\MT(\mathcal{A}_W) \gets \emptyset$
    \State $\mathcal{M} \gets \{p_i p_j, ~ 1 \leqslant i,j \leqslant D\}$.
    \State Order $\mathcal{M}$ using $<_X$
    \For{$m \in \mathcal{M}$ of total degree $0$ or $1$}
        \State $\MT(\mathcal{A}_W) \gets \MT(\mathcal{A}_W) \cup {[m]}$

    \EndFor

    \For{$m \in \mathcal{M}$ of total degree greater or equal than 2}
        \State Write $m = X_j m'$ with $\deg(m') < \deg(m)$
        \State $\overrightarrow{[m]} \gets M_{X_j} \cdot \overrightarrow{[m']}$
        \State $\MT(\mathcal{A}_W) \gets \MT(\mathcal{A}_W) \cup {[m]}$
    \EndFor

    \noindent
    \Return{$\MT(\mathcal{A}_W)$}
    \end{algorithmic}
\end{algorithm}

\begin{algorithm}[H]
\caption{CharacteristicPolynomial}
\begin{algorithmic}
\Require $\MT \left(\mathcal{A}_W\right)$ ; $M_t$ the multiplication matrix by $t$ in $\mathcal{A}_W$ ; $\VTr 1$.
\Ensure $\chi_t$ the characteristic polynomial of $t$.
\State $N_0 \gets D$.
\State $\overrightarrow{v} \gets [1, 0, \dots, 0]^T$.
\For{$i = 1, \dots, D$}
    \State $\overrightarrow{v} \gets M_t \cdot \overrightarrow{v}$.
    \State $N_i \gets \Trace (t^i) = \overrightarrow{v} \cdot \VTr(1)$.
\EndFor
\State Solve the triangular system $\big\{ (D-k)b_k = \sum_{j=0}^k b_{k-j} N_j\big\}_{k \in [\![1, D]\!]}$ whose unknowns are the $b_k$.

\noindent
\Return{$\chi_t  = T^D + \sum_{k=0}^{D-1} b_k U_0^{D-k}$.}
\end{algorithmic}
\end{algorithm}

\begin{algorithm}[H]
\caption{RURPolynomials}
\begin{algorithmic}
\Require $\chi_t = \sum_{i=0}^D b_i T^{D-i}$ ; $M_{X_1}, \dots, M_{X_n}$ the multiplication matrices of $X_1, \dots, X_n$ in $\mathcal{A}_W$ ; $\VTr 1$.
\Ensure $h_{X_1}, \dots, h_{X_n}$.
\For{$i = 1, \dots, D$}
	\State $H_i (T) \gets \sum_{j=0}^i b_j T^{i-j}$.
\EndFor
\For{$k = 1, \dots, n$}
	\For{$i = 0, \dots, D-1$}
		\State $\Trace(X_kt^i) \gets \overrightarrow{X_k t^i} \cdot \VTr 1$.
	\EndFor
	\State $h_{X_k} \gets \frac{\sum_{i=0}^{D-1} \Trace(X_k t^i) H_{D-1-i} (T)}{\frac{\partial \chi_t}{\partial T}}$.
\EndFor

\noindent
\Return{$h_{X_1}, \dots, h_{X_n}$.}
\end{algorithmic}
\end{algorithm}

\begin{algorithm}[H]
\caption{SeparationCheck}
\begin{algorithmic}
\Require A linear form $t = a_1 X_1 + \dots + a_n X_n \in \C[X]$ ; $\chi_t$ the characteristic polynomial of $t$ in $\mathcal{A}_W$ ; $\mathcal{G}$ a reduced Gröbner basis of $\mathcal{I}_W$ for $<_X$.
\Ensure True if $t$ is generically separating and False otherwise.
\State Compute $\Disc(\chi_t)$
\If{$Disc(\chi_t) = 0$}
    \Return(False)
\EndIf
\State Choose $w \in \C^s$ such that $w \notin \cup_{g \in \mathcal{G}} V(\LC_{<X})$ ; $V(\mathcal{I}(w))$ has dimension zero and $\Disc(\chi_t) (w) \neq 0$
\State Use \cite[Algorithm 6]{DRR2025} with $\mathcal{G}(w)$ to know if $t$ is separating

\noindent
\Return{The answer given}
\end{algorithmic}
\end{algorithm}

%
%

\end{document}